\documentclass[12pt,fleqn]{article}
\usepackage[T1]{fontenc}
\usepackage{amssymb, amsmath, graphicx, subfigure, xparse}
\usepackage{float}
\usepackage{physics}
\usepackage{fullpage}
\usepackage[hidelinks]{hyperref}
\usepackage{qcircuit}
\usepackage{color}
\usepackage{mdframed}
\usepackage{mathpazo}

\usepackage{todonotes}

\usepackage{thmtools}
\usepackage{thm-restate}

\usepackage{complexity}

\usepackage{tikz}
\usetikzlibrary{shapes,backgrounds}

\newtheorem{theorem}{Theorem}
\newtheorem{definition}[theorem]{Definition}

\newtheorem{lemma}[theorem]{Lemma}
\newtheorem{conjecture}[theorem]{Conjecture}
\newtheorem{corollary}[theorem]{Corollary}

\newtheorem{claim}[theorem]{Claim}

\newtheorem{fact}[theorem]{Fact}

\usepackage{amsmath, amssymb, graphicx, subfigure}
\usepackage{enumitem}
\usepackage{xparse}
\usepackage{physics}
\usepackage{color}

\renewcommand{\tilde}{\widetilde}

\newcommand{\CC}{\mathbb{C}}

\newcommand{\II}{\mathbb{I}}

\newcommand{\Cc}{\mathcal{C}}

\newcommand{\Ee}{\mathcal{E}}
\newcommand{\cE}{\mathcal{E}}
\newcommand{\cH}{\mathcal{H}}

\newcommand{\Hh}{\mathcal{H}}
\newcommand{\Ll}{\mathcal{L}}

\newcommand{\Pp}{\mathcal{P}}
\newcommand{\Ss}{\mathcal{S}}
\renewcommand{\cS}{\Ss}

\newcommand{\Uu}{\mathcal{U}}

\newcommand{\Exp}{\mathop{\mathbb{E}\hspace{0.13em}}}

\newcommand{\inv}[1]{#1^{-1}}

\newcommand{\id}{\mathrm{id}}

\newcommand{\eps}{\epsilon}

\newcommand{\ent}{\mathrm{S}} 

\newcommand{\defeq}{\mathrel{\overset{\makebox[0pt]{\mbox{\normalfont\tiny\sffamily def}}}{=}}}

\usepackage{complexity}

\renewcommand{\exp}{\mathsf{exp}}

\newcommand{\bits}{\{0,1\}}

\newcommand{\Enc}{\mathrm{Enc}}

\newcommand{\cat}{\mathsf{cat}}

\newcommand{\anc}{\mathsf{anc}}
\newcommand{\code}{\mathsf{code}}

\newcommand{\qubit}{\CC^{2}}
\newcommand{\qubits}[1]{(\qubit)^{\otimes #1}}

\newcommand{\supp}{\mathrm{supp}}

\newcommand{\dgre}{\mathsf{deg}}
\newcommand{\trunc}{\mathsf{trunc}}

\renewcommand{\Tr}{\tr}

\newcommand{\depth}{\mathrm{depth}}
\newcommand{\cd}{\mathsf{cc}}

\newcommand{\br}[1]{\left(#1\right)}

\renewcommand{\id}{\mathbb{I}}
\renewcommand{\exp}{\mathrm{exp}}

\newenvironment{proof}{\noindent{\bf Proof:} \hspace*{1mm}}{
    \hspace*{\fill} $\Box$ \\ }
\newenvironment{proof_of}[1]{\noindent {\bf Proof of #1:}
    \hspace*{1mm}}{\hspace*{\fill} $\Box$ \\}
\newenvironment{xalign}{\subequations\align}{\endalign\endsubequations}

\makeatletter
\renewcommand{\maketitle}{\bgroup\setlength{\parindent}{0pt}
\begin{flushleft}
  \LARGE \textbf{\@title} \linebreak
  
  \normalsize \@author \linebreak
\end{flushleft}\egroup
}
\makeatother

\title{Circuit lower bounds for low-energy states of \\ quantum code Hamiltonians}

\usepackage{authblk}

\author[1,2]{Anurag Anshu}
\author[2,3]{Chinmay Nirkhe}

\affil[1]{\small Simons Institute for the Theory of Computing, Berkeley, California, USA 94720}
\affil[2]{\small Challenge Institute for Quantum Computation, University of California, Berkeley 94720}
\affil[3]{\small Electrical Engineering and Computer Sciences, University of California, Berkeley 94720}
\affil[ ]{\small \texttt{\{anuraganshu,nirkhe\}@berkeley.edu}}

\date{}

\begin{document}

\maketitle{}

\begin{abstract}
The No Low-energy Trivial States (NLTS) conjecture of Freedman and Hastings \cite{10.5555/2600498.2600507} --- which posits the existence of a local Hamiltonian with a super-constant quantum circuit lower bound on the complexity of all low-energy states --- identifies a fundamental obstacle to the resolution of the quantum PCP conjecture.
In this work, we provide new techniques, based on entropic and local indistinguishability arguments, that prove circuit lower bounds for all the low-energy states of local Hamiltonians arising from quantum error-correcting codes. \\

For local Hamiltonians arising from nearly linear-rate or nearly linear-distance LDPC stabilizer codes, we prove super-constant circuit lower bounds for the complexity of all states of energy $o(n)$. Such codes are known to exist and are not necessarily locally-testable, a property previously suspected to be essential for the NLTS conjecture. Curiously, such codes can also be constructed on a two-dimensional lattice, showing that low-depth states cannot accurately approximate the ground-energy even in physically relevant systems. 
\end{abstract}

\pagebreak

\section{Introduction}

Ground- and low-energy states of local Hamiltonians are the central objects of study in condensed matter physics. The $\QMA$-complete local Hamiltonian problem is also the quantum analog of the $\NP$-complete constraint satisfaction problem (CSP) and ground-states (and low-energy states) of local Hamiltonians correspond to solutions (near optimal solutions) of the problem. A sweeping insight into the computational properties of the low energy spectrum is embodied in the quantum PCP conjecture, which is arguably the most important open question in quantum complexity theory. Just as the classical PCP theorem establishes that CSPs with a promise gap remain $\NP$-complete, the quantum PCP conjecture asserts that local Hamiltonians with a promise gap remain $\QMA$-complete. But despite numerous results providing evidence both for \cite{10.1145/1536414.1536472,10.5555/2600498.2600507,8104078,nirkhe_et_al:LIPIcs:2018:9095,DBLP:conf/focs/NatarajanV18,Eldar2019RobustQE} and against \cite{10.5555/2011637.2011639,10.1145/2488608.2488719,10.1007/s11128-014-0877-9} the quantum PCP conjecture, the problem has remained open for nearly two decades.

The difficulty of the quantum PCP conjecture has motivated a flurry of research beginning with Freedman and Hastings' \emph{No low-energy trivial states (NLTS) conjecture} \cite{10.5555/2600498.2600507}. The NLTS conjecture posits that there exists a fixed constant $\eps > 0$ and a family of $n$ qubit local Hamiltonians such that every state of energy $\leq n\eps$ requires a quantum circuit of super-constant depth to generate. The NLTS conjecture is a necessary consequence of the quantum PCP conjecture because $\QMA$-complete problems do not have $\NP$ solutions and a constant-depth quantum circuit generating a low-energy state would serve as a $\NP$ witness. Thus, this conjecture addresses the inapproximability of local Hamiltonians by classical means.  
Proving the NLTS conjecture remains a fundamental obstacle to the resolution the quantum PCP conjecture.

In this work, we show that for local Hamiltonians corresponding to LDPC stabilizer quantum error-correcting codes of linear rate and polynomial distance, every state of energy $\leq \eps n$ requires a quantum circuit of depth $\Omega(\log 1/\eps)$ to generate. We also show similar results for linear distance LDPC codes. Thus, any improvement to our result would resolve the NLTS conjecture.

\subsection{Our results}

We restrict our attention to quantum error-correcting codes and the low-energy states of the associated code Hamiltonians\footnote{The classical analog of this question, the circuit complexity of approximate sampling from the uniform distribution of a classical error-correcting code, is answered by Lovett and Viola \cite{lovett-viola}.}. A code Hamiltonian is a local Hamiltonian whose ground-space is precisely the code-space, with the additional property that the energy of a state measures the number of violated code checks. 
Examples of quantum error-correcting codes realized as the ground-spaces of local Hamiltonians already play a central role in our understanding of the physical phenomenon known as topological order \cite{KITAEV20032,PhysRevLett.107.210501}. 
Call an error-correcting code an $[[n,k,d]]$ code with locality $\ell$ if it has $n$ physical qubits, $k$ logical qubits, distance $d$ and the corresponding code Hamiltonian has locality $\ell$ (these definitions are made precise in Section \ref{sec:preliminaries}).
Our main result refers to a subclass of codes known as \emph{stabilizer codes} where the code Hamiltonian is commuting and each Hamiltonian term is the tensor product of Pauli operators.

\begin{theorem}
\label{thm:main}
Let $\Cc$ be a $[[n,k,d]]$ stabilizer code of constant locality $\ell=O(1)$ and let $H=\sum_i H_i$ be the corresponding code Hamiltonian with a term $H_i = (\id - C_i)/2$ for each code check $C_i$. For any $\eps>0$ and any state $\psi$ on $n$-qubits with energy $\leq\eps n$, the circuit depth of $\psi$ is at least
\begin{align}
    \Omega \left( \min \left\{ \log d, \quad \log {\frac{k + d}{n\sqrt{\eps \log\frac{1}{\eps}}}} \right\} \right).
\end{align}
\end{theorem}

In the case of linear-rate and polynomial-distance codes such as the hypergraph product code of Tillich and Z\'emor \cite{5205648}, the theorem proves a circuit lower bound of $\Omega(\delta \log n)$ for any state of energy $O(n^{1-\delta})$. So for fixed $\delta$, say $\delta = 0.01$, it provides a circuit lower bound of $\Omega(\log n)$ for all states of energy $O(n^{0.99})$. Furthermore, it proves a circuit lower bound of $\Omega(\log \log n)$ for any state of energy $O(n/\poly \log n)$ and a super-constant circuit lower bound for any state of energy $o(n)$. {Recent developments \cite{panteleev2020quantum, hastings2020fiber,Breuckmann2020BalancedPQ} have shown quantum LDPC codes with near-linear distance $d=\Omega(n/\log n)$  (but with low rate $k=O(1)$). The theorem also provides a circuit lower bound of $\Omega(\log n)$ for all states of energy $O(n^{0.99})$ in such codes.}
For ``reasonable'' stabilizer codes of polynomial rate and polynomial distance, this theorem provides a non-trivial lower bound on the circuit complexity in the energy regime of $1/\poly(n)$. 

Furthermore, for any stabilizer code of nearly-linear-rate (i.e. $n^{1-\delta}$ rate) and distance at least $n^{\Omega(\delta)}$, the theorem still proves a circuit lower bound of $\Omega(\delta \log n)$ for any state of energy $O(n^{1-2\delta})$. Codes with these properties are known to exist on constant-dimensional lattices and are not locally-testable; one example is the punctured 2D toric code\footnote{The punctured 2D toric code is known to saturate the information-distance tradeoff bound of \cite{PhysRevLett.104.050503}.} with $O(n^{1-\delta})$ punctures \cite{PhysRevLett.104.050503,PhysRevA.86.032324}. Additionally, toric codes defined on hyperbolic manifolds where the manifold has constant negative curvature also have linear rate and small (yet polynomial) distance. Examples include the toric code defined on  4-dimensional arithmetic hyperbolic manifolds \cite{doi:10.1063/1.4891487} or golden codes \cite{10.5555/3370251.3370252}, for which our main result will also prove a super-constant circuit lower bound for all states of energy $o(n)$. 

\subsection{Challenges and an overview of proof techniques}

Code-states of an error correcting code are well known to have a large circuit complexity $\sim\log d$, where $d$ is the distance of the code. This lower bound arises from the local indistinguishability property (see Fact \ref{fact:local-indistinguishability}), which means that for any size $< d$ subset $S$ of the qubits, the reduced density matrix $\rho_S$ for any code-state $\rho$ is an invariant of the code-space.

A natural notion of approximation to code-states is the class of low-error states. Such states resemble the code-states on a large number of physical qubits, differing arbitrarily on a small fraction (interpreted as an error). Prior works \cite{8104078, nirkhe_et_al:LIPIcs:2018:9095}, exploiting the error-correction property, showed that the low-error states also have a large circuit complexity. This generalized the aforementioned circuit lower bounds on code-states. However, as further demonstrated in \cite{nirkhe_et_al:LIPIcs:2018:9095}, low-error is a strictly weaker notion than low-energy . Without invoking highly non-trivial properties such as local testability \cite{8104078}, it seems unclear if the low-energy states can be viewed as low-error. This leads to the central challenge towards the NLTS conjecture: capturing the circuit complexity of the low-energy states. The prior arguments, all of which rely on local indistinguishability (captured by the code distance), do not seem to suffice.

We observe, for the first time, that another parameter plays a key role in circuit lower bounds: the rate of the code. Inspired by \cite{PhysRevLett.104.050503}, we use novel entropic arguments to prove that states of low circuit complexity are significantly far in $\ell_1-$distance from high rate code-spaces (established in Section \ref{sec:warmup-entropy}). Formally, we show that all states of circuit complexity $\leq \log d$ are at a $\ell_1$-distance of $\geq \Omega(\frac{k^2}{n^2})$ from the code-space\footnote{
This is proved using an information theory argument. Consider a state $\psi$ with small trace distance to the code. Then, the reduced density matrices $\{\psi_S\}$ approximate the reduced density matrices of the closest state of $\Cc$. By local indistinguishability, the $\{\psi_S\}$ in turn approximate the reduced density matrices for all code-states. In particular, they approximate the reduced density matrices of the encoded maximally-mixed state $\Theta$ of the code. This state has entropy $S(\Theta)$ equal to the rate of the code, $k$. We now show that if $\psi$ has low circuit complexity, then the entropy $S(\Theta)$ is bounded. 
Assume that $\psi$ is the output of a low-depth circuit $W$, then for any qubit $i$,
\begin{align}
    \tr_{-\{i\}}(W^\dagger \psi W) \approx \tr_{-\{i\}}(W^\dagger \Theta W). \label{eq:approxidea}
\end{align}
This is because (a) $\tr_{-L_i}(\psi) \approx \tr_{-L_i}(\Theta)$ where $L_i$ is the support of the lightcone of qubit $i$ with respect to $W$ and (b) the value of the $i$th qubit of a $W$-rotated state only depends on the lightcone of the $i$th qubit. However, the left-hand side of \eqref{eq:approxidea} equals the pure state $\ketbra{0}{0}$ and so the entropy of $\tr_{-\{i\}}(W^\dagger \Theta W)$, the $i$th qubit of $W^\dagger \Theta W$, is small. This gives us an overall bound on the entropy of $W^\dagger\Theta W$ which equals that of $\Theta$ and also upper bounds the rate of the code.
}.

This observation alone does not suffice to address the aforementioned central challenge: the space of low-energy states is much larger than the code-space or even its small neighborhood. A general strategy in earlier works \cite{8104078,Eldar2019RobustQE} was to build a low-depth decoding circuit to bring each low-energy state closer to the code-space. But this required assuming that the code was locally testable; such codes are not known to exist in the desired parameter regime. We instead appeal to the observation that every eigenspace of a stabilizer code Hamiltonian possesses the local indistinguishability property (Fact \ref{fact:local-indistinguishability}). Instead of attempting to construct a decoding circuit, we measure the syndrome using a constant-depth circuit (which uses the LDPC nature of the code Hamiltonian). This allows us to decohere the low energy state into a mixture of orthogonal states that live within each of the eigenspaces. A key realization is that measurement of the syndrome for low-energy states is a gentle measurement in that it does not perturb the state locally. This is used to show that a state of low energy satisfies an approximate version of locally indistinguishability. This, coupled with the argument for codes of high rate, completes the proof.

\subsection{Separation of the NLTS conjecture from the QLDPC/QLTC conjectures}
A quantum low-density parity-check (LDPC) code is an error-correcting code with a local Hamiltonian defining the code-space, such that each qubit participates in at most a constant number of Hamiltonian terms and each Hamiltonian term acts on at most a constant number of qubits (i.e. the bipartite interaction matrix has low-density). The QLDPC conjecture posits the existence of LDPC codes that also have linear-rate and linear-distance. It has been previously suspected that a QLDPC property would be necessary for NLTS Hamiltonians \cite{10.5555/2600498.2600507,doi:10.1137/140975498,hastings17,8104078,nirkhe_et_al:LIPIcs:2018:9095}. Our result breaks this intuition by showing that lower bound results are achievable even when the distance is a small polynomial; interestingly, it is the rate that needs to be almost linear for our result, a counter-intuitive property. Furthermore, our results show that entanglement persists at energy well past the distance threshold; 
a regime where one intuitively expects the stored information to be lost.
Furthermore, it is believed that the QLDPC codes also need to be locally-testable \cite{doi:10.1137/140975498} for NLTS. This fact is formalized by Eldar and Harrow \cite{8104078} who give a construction of an NLTS Hamiltonian from any locally-testable CSS QLDPC code with constant soundness. Quantum locally testable codes (QLTCs) of constant soundness are not known to exist; the best constructions achieve a soundness factor of $O(1/\poly \log n)$  with a distance of $\Omega(\sqrt{n})$ \cite{hastings17,1911.03069}. Our construction does not require local-testability; in fact, the hypergraph product code \cite{5205648} with linear rate and polynomial distance is not locally-testable as there are errors of size $\Omega(\sqrt{n})$ that violate only a single check \cite[page 4]{1911.03069}.

\subsection{Spatially local Hamiltonians}
A key property of an NLTS Hamiltonian is that it cannot live on a lattice of dimension $D$ for a fixed constant $D$ \cite{10.1145/2491533.2491549}. This is because of a ``cutting'' argument: Let $H$ be a local Hamiltonian in $D$ dimensions and $\Psi$ a ground-state of $H$.
For a fixed constant $\eps$, partition the lattice into $D$ dimensional rectangular chunks so that the side length of each rectangular chunk is $O((D\eps)^{-1/D})$. Let $\rho_i$ be the reduced state of $\Psi$ on a chunk $i$, and $\rho = \bigotimes_i \rho_i$ be a state over all the qubits. It's not hard to check that $\rho$ violates at most a $\eps$-fraction of the terms of $H$ (only the boundary terms of the rectangular division) and yet has circuit complexity at most $\exp(((D \eps)^{-1/D} )^{D}) = O(\exp(1/D\eps)) = O(1)$; so it is not NLTS. 

This circuit complexity upper bound can be further improved for the specific case of stabilizer Hamiltonians on a lattice, due to the result of Aaronson and Gottesman \cite{PhysRevA.70.052328}. Since the circuit complexity of each chunk is at most logarithmic in its size $O(1/\eps^{1/D})$, the aforementioned quantum state $\rho$ can actually be prepared by a circuit of depth $O(\min (\log n, \log(1/\eps)))$. Note that this holds for any $0<\eps < 1$, not just a constant. Therefore, our lower bound in the case of nearly linear rate and polynomial distance codes (such as the punctured toric code) matches the upper bound -- up to constant factors -- closing the question on the circuit complexity of the approximate ground-states of these codes.

We also highlight that the only known constructions of LDPC stabilizer codes of linear rate and polynomial distance are built from classical expander graphs and therefore cannot live on a lattice of constant dimension $D$. Therefore, our result in Theorem \ref{thm:main} (applied to linear rate codes) conveniently evades this counterexample.

\subsection{The physics perspective}

The crucial role of entanglement in the theory of quantum many-body systems is widely known with some seminal examples including topological phases of matter \cite{KitaevP06} and quantum computation with physically realistic systems \cite{RaussendorfB01, RaussendorfB03}. But entanglement also brings new challenges as the classical simulation of realistic many-body systems faces serious computational overheads. 

Estimating the ground-energy of such systems is one of the major problems in condensed matter physics \cite{White92}, quantum chemistry \cite{Cao2019}, and quantum annealing \cite{ApolloniCF89, KadowakiN98}. One of the key methods to address this problem is to construct \emph{ansatz quantum states} that achieve as low-energy as possible and are also suitable for numerical simulations. A leading ansatz, used in Variational Quantum Eigensolvers \cite{Peruzzo14, LaRose2019, Cao2019} or Quantum Adiabatic Optimization Algorithm \cite{farhi2014quantum}, is precisely the class of quantum states that can be generated by low-depth quantum circuits. 

Theorem \ref{thm:main} shows that there are Hamiltonians for which any constant-depth ansatz cannot estimate their ground-energies beyond a fairly large threshold. As discussed earlier, we provide examples even in the physically realistic two-dimensional setting. For example, the 2D punctured toric code Hamiltonians on $n$ qubits with distance $d$ (which is a free parameter) requires a circuit of depth $\Omega(\log d)$ for an approximation to ground-energy better than $O(n/d^3)$.  

\subsection{Prior Results}

To the best of our knowledge, prior to this result, a circuit lower bound on the complexity of \emph{all} low-energy states was only known for states of energy $O(n^{-2})$. This result follows from the $\QMA$-completeness of the local Hamiltonian problem with a promise gap of $O(n^{-2})$ (assuming $\NP \neq \QMA$); the original proof of Kitaev had a promise gap of $O(n^{-3})$  \cite{10.5555/863284,doi:10.1137/S0097539704445226} which was improved by \cite{PhysRevA.97.062306,Bausch2018analysislimitations}.

\paragraph{Robustness to perturbations}

Prior to our results, being robust to constant-distance perturbations was only known in the special case of CSS codes of distance $\omega(n^{0.5})$ or larger \cite{8104078}, as exhibited in recent codes \cite{hastings2020fiber, panteleev2020quantum} . We prove as a warm-up that the robust notion of circuit complexity holds (a) for any state of a quantum code of linear-rate (Lemmas \ref{lem:warmup-lemma} and \ref{lem:rootdwarmup}) or (b) any state of a quantum code of linear-distance (Lemma \ref{lem:lineardist}). For commuting codes, we can further show that they are robust to perturbations in trace distance very close to $1$ (Lemma \ref{lem:rootdwarmup}).

\paragraph{Subclasses of low-energy states}
Freedman and Hastings proved a circuit lower bound for all ``one-sided'' low-energy states of a particular stabilizer Hamiltonian where a state is one-sided if it only violates either type $X$ or type $Z$ stabilizer terms but not both \cite{10.5555/2600498.2600507}. 

As mentioned earlier, a different line of work focused on the low-error states \cite{8104078}, which differ from a code-state on at most $\eps n$ qubits. These works prove a circuit lower bound of $\Omega(\log n)$ on the complexity of all low-error states of a specific local Hamiltonian \cite{8104078,nirkhe_et_al:LIPIcs:2018:9095} (for some constant $\eps$). %

Eldar has also shown an $\Omega(\log n)$ circuit lower bounds for Gibbs or thermal states of local Hamiltonians at $O(1/\log^2 \log n)$ temperature \cite{Eldar2019RobustQE} which is a specific low-energy state formed by coupling the ground-state of a physical system to a ``heat bath.''

Bravyi \emph{et. al.} \cite{BravyiKKT19} give examples of classical Hamiltonians for which all the states with $\mathbb{Z}_2$ symmetry require $\Omega(\log n)$ circuit complexity. 

\subsection{Future work and open questions}

Our main result, Theorem \ref{thm:main}, comes quite close to proving the NLTS conjecture and, therefore, provides some of the first positive results for the QPCP conjecture in a long time. And while we are able to provide a constant depth lower bound for all states of energy $\leq n/100$, we have been unsuccessful, thus far, at extending this result to a super constant depth lower bound. Below are sketches of some potential avenues at completing the proof of the NLTS result and future questions to consider.

\subsubsection{Gap amplification implies circuit complexity amplification}

\label{sec:openamplify}

Consider a transformation of a local Hamiltonian $H$ into a different Hamiltonian $H'$ such that for all low-energy states $\phi$, $\tr(H' \phi) \geq p \cdot \tr(H \phi)$ for a large value $p$. If one could perform this transformation without increasing the locality of the Hamiltonian or the norm of the Hamiltonian for $p = \omega(1)$, we would obtain an NLTS Hamiltonian due to Theorem \ref{thm:main}.  
However, we do not know any construction of such a transformation. The closest result we know is a construction that amplifies the energy of all low-depth states at the cost of increasing the locality of the Hamiltonian;  we state this result as Theorem \ref{thm:amplification} in Appendix \ref{sec:amplification}, where we delve into this idea in greater detail. This is analogous to the amplification step in Dinur's PCP theorem \cite{dinur-pcp} or the quantum gap amplification lemma \cite{10.1145/1536414.1536472} which performs amplification for the lowest energy eigenstate. In order to complete the transformation, we need a ``locality reduction'' transformation which reduces the locality of the Hamiltonian while preserving the energy of all low-depth states. The analogous transformation from Dinur's PCP theorem involves cloning of information, which is not viable in the quantum context.

\subsubsection{Improving our circuit lower bound, Theorem \ref{thm:main}}

Our results give the strongest circuit lower bounds when we consider stabilizer codes of linear rate and polynomial distance. There are two possibilities to consider: (1) such codes are simply not necessarily NLTS and another assumption is needed to prove NLTS or (2) our proof techniques have the potential to be strengthened in order to prove NLTS. For the first possibility, recall that Eldar and Harrow \cite{8104078} proved that locally-testable linear distance CSS codes are NLTS; however, such codes are not known to exist. For the second possibility, let us reexamine our proof in greater details. 

At a high level, we proceed via a proof-by-contradiction in which we assume that a low-depth state $\ket \psi$ of low-energy exists. We then construct a related state $\Theta$ whose entropy is at least $k$ (the rate of the code) and yet its entropy is upper bounded by $O(2^t \eps n)$ due to its similarity to $\ket \psi$ on small marginals. The $\eps n$ term in the upper bound is a consequence of the gentle measurement lemma. %
The $2^t$ term is, we believe, an over-counting due to the na\"ive application of sub-additivity of entropy to the state $\Theta$. Instead, there may be a more sophisticated analysis using conditional entropies, which would avoid the $2^t$ factor. If successful, this would prove NLTS for some constant $\eps = \Omega(k/n) = \Omega(1)$.

\subsubsection{Lower bounds against other classical approximations}

The NLTS conjecture postulates a lower bound on the approximability of local Hamiltonians by low-depth circuits. However, if the QPCP conjecture is true and $\NP \neq \QMA$, then we should be able to prove lower bounds on the approximability of local Hamiltonians by any means of classically simulation, such as stabilizer codes, tensor networks of low tree width, etc. Our result, of course, does not prove general lower bounds against classical simulation as the Hamiltonians are stabilizer codes, which can be solved classically. One interpretation of our result is a proof that stabilizer codes and low-depth circuits are not equal in simulation power (even in an approximate sense). One future avenue of research is to produce similar lower bounds against other methods of classical simulation. Furthermore, one should be able to provide lower bounds against generalized non-deterministic computation: show that there exists a family of local Hamiltonians which has no low-energy states whose energy can be verified in $\textsf{NPTIME}[t(n)]$ for progressively larger and larger $t(n)$ up to $\poly(n)$.

\subsection*{Organization}
In Section \ref{sec:preliminaries}, we give the necessary background information. Section \ref{sec:warmup-entropy} proves the entropy-based robust lower bound for code-states of any code (Lemma \ref{lem:warmup-lemma}). In Section \ref{sec:main}, we prove our main result, Theorem \ref{thm:main}. Appendix \ref{sec:other-techniques} contains proofs of other related techniques for lower bounding circuit complexity. Appendix \ref{sec:appendix-pf} contains improved lower bound techniques based on approximate ground-state projectors. Appendix \ref{sec:amplification} contains the mathematics behind the gap amplification techniques suggested in Section \ref{sec:openamplify}.

\section{Preliminaries}
\label{sec:preliminaries}

We will assume that the reader is familiar with the basics of quantum computing and quantum information.

\subsection{Notation}
\label{subsec:notation}

The set of integers $\{1,2,\ldots m\}$ is abbreviated as $[m]$. Given a composite system of $m$ qubits, we will often omit the register symbol from the states (being clear from context). For a set $A\subseteq [m]$, $-A$ will denote the set complement $[m] \setminus A$ and $\tr_A$ will denote the partial trace operation on qubits in $A$ and $\tr_{-A}\defeq\tr_{[m]\setminus A}$. Therefore, $\tr_{-\{i\}}(\cdot)$ gives the reduced marginal on the $i$th qubit.

\paragraph{Quantum states} A quantum state is a positive semi-definite matrix with unit trace, acting on a finite-dimensional complex vector space (a Hilbert space) $\cH$. In this paper we will only concern ourselves with Hilbert spaces coming from a collection of qubits, i.e. $\cH = (\CC^2)^{\otimes m}$. A pure quantum state is a quantum state with rank 1 (i.e. it can be expressed as $\ketbra{\psi}{\psi}$ for some unit vector $\ket{\psi}$). In which case we will refer to the state as $\ket{\psi}$ when interested in the unit vector representation and $\psi$ when interested in the positive semi-definite matrix representation. Given two Hilbert spaces $\cH_A,\cH_B$, their tensor product is denoted by $\cH_A\otimes \cH_B$. For a quantum state $\rho_{AB}$ acting on $\cH_A\otimes \cH_B$, the reduced state on $\cH_A$ is denoted by $\rho_A\defeq \tr_{B}(\rho_{AB})$, where $\tr_B$ is the partial trace operation on the Hilbert space $\cH_B$. The partial trace operation is a type of quantum channel. More generally, a quantum channel $\cE$ maps quantum states acting on some Hilbert space $\cH_A$ to another Hilbert space $\cH_B$. 

Every quantum state $\rho$ acting on a $D$-dimensional Hilbert space has a collection of eigenvalues $\{\lambda_i\}_{i=1}^D$, where $\sum_i\lambda_i=1$ and $\lambda_i\geq 0$. The von Neumann entropy of $\rho$, denoted $\ent(\rho)$, is defined as $\sum_i \lambda_i \log\frac{1}{\lambda_i}$. All logarithms are in base 2.

Unless specified otherwise, assume that we are considering a quantum code on $n$ physical qubits and assume we are considering quantum states on an expanded Hilbert space of $m \geq n$ qubits. We will denote the $n$ qubits corresponding to the code-space as $\code$ and the remainder $(m - n)$ qubits defining the expanded Hilbert space as $\anc$ for ancillas. Furthermore, the reduced density matrix on $\code$ of a state $\rho$ will be referred to as $\rho_\code$ and, respectively, $\rho_\anc$ for the ancillas. The uniformly distributed quantum state on a Hilbert space $\Hh$ will be represented by $\nu$:
\begin{align}
    \nu_\Hh \defeq \frac{\id_\Hh}{\abs{\Hh}}.
\end{align}

\subsection{Error-correction}  

Here we recall the definitions of a quantum error-correcting code. We will refer to a code $\Cc$ as a $[[n,k,d]]$ code where $n$ is the number of physical qubits (i.e. the states are elements of $\qubits{n}$), $k$ is the dimension of the code-space, and $d$ is the distance of the code. We can define distance precisely using the Knill-Laflamme conditions \cite{PhysRevA.55.900}.

Let $\{\ket{\overline{x}}\} \subseteq \Cc$ be an orthonormal basis for $\Cc$ parameterized by $x \in \bits^k$. The Knill-Laflamme conditions state that the code can correct an error $E$ iff
\begin{align}
\bra{\overline x}E \ket{\overline y} = \begin{cases} 0 & x \neq y \\
\eta_E & x = y \end{cases}
\end{align}
where $\eta_E$ is a constant dependent on $E$. This is equivalent to
\begin{align}
\Pi_\Cc E \Pi_\Cc = \eta_E \Pi_\Cc \label{eq:kl-conditions}
\end{align}
where $\Pi_\Cc$ is the projector onto the code-space. We say that the code $\Cc$ has distance $d$ if it can correct all Pauli-errors of weight $<d$. By linearity, it can then correct all errors of weight $<d$. Furthermore, given a set $S$ of fewer than $d$ qubits, the reduced density matrix $\rho_S$ of any code-state $\rho$ on the set $S$ is an invariant of the code. Intuitively, this property can be seen as a consequence of the no-cloning theorem since $\rho$ can be recovered exactly from $\rho_{-S}$; therefore, $\rho_S$ cannot depend on $\rho$ without violating the no-cloning theorem. The property can also be derived as a direct consequence of the Knill-Laflamme conditions and is known as \emph{local indistinguishability}.

\begin{fact}[Local Indistinguishability]
\label{fact:local-indistinguishability}
Let $\Cc$ be a $[[n,k,d]]$ error correcting code and $S$ a subset of the qubits such that $\abs{S} < d$. Then the reduced density matrix $\rho_S$ of any code-state $\rho$ on the set $S$ is an invariant of the code.
\end{fact}

\begin{proof}
Let $E$ be any operator whose support is entirely contained in $S$. Then for any code-state $\rho$,
\begin{xalign}
\tr(E\rho) &= \tr(E \Pi_\Cc \rho \Pi_\Cc) \label{eq:li-codedef1} \\
&= \tr(\Pi_\Cc E \Pi_\Cc \rho) \label{eq:li-ct} \\
&= \tr(\eta_E \Pi_\Cc \rho) \label{eq:li-codedef2} \\
&= \eta_E \label{eq:li-tr1}
\end{xalign}
where \eqref{eq:li-codedef1} is because $\rho$ is a code-state, \eqref{eq:li-ct} is due to cyclicality of trace, \eqref{eq:li-codedef2} is an application of \eqref{eq:kl-conditions}, and \eqref{eq:li-tr1} is because $\rho$ has trace 1. Since this equality holds for any operator $E$ and $\eta_E$ is a constant independent of $\rho$ such that $\eta_E = \tr(E\rho) = \tr(E \rho_S)$, then $\rho_S$ is an invariant of the code-state $\rho$.
\end{proof}

Given a code $\Cc$ and a state $\sigma$ on $n$ qubits, we define the trace-distance between $\sigma$ and $\Cc$ as $    \inf_{\rho \in \Cc} \norm{\rho - \sigma}_1.$

We will refer to a code $\Cc$ as a \emph{stabilizer} code if it can be expressed as the simultaneous eigenspace of a subgroup of Pauli operators.

\begin{definition}[Pauli group]
The Pauli group on $n$ qubits, denoted by $\Pp_n$ is the group generated by the $n$-fold tensor product of the Pauli matrices
\begin{align}
    \II_2 = \begin{pmatrix} 1 & 0 \\ 0 & 1\end{pmatrix}, \quad X = \begin{pmatrix} 0 & 1 \\ 1 & 0 \end{pmatrix}, \quad Y = \begin{pmatrix} 0 & -i \\i & 0\end{pmatrix} \textrm{, and }\quad Z = \begin{pmatrix} 1 & 0 \\ 0 & -1\end{pmatrix}.
\end{align}
\end{definition}

\begin{definition}[Stabilizer Code]
Let $\{C_i\}_{i \in [N]}$ be a collection of commuting Pauli operators from $\Pp_n$ and $\Ss$ be the group generated by $\{C_i\}$ with multiplication. The stabilizer error-correcting code $\Cc$ is defined as the simultaneous $+1$ eigenspace of each element of $\Ss$: 
\begin{align}
    \Cc = \left\{ \ket{\psi} \in \qubits{n} : C_i \ket{\psi} = \ket{\psi} \ \forall i \in [N] \right\}.
\end{align}
More generally, for every $s \in \bits^N$, define the space $D_s$ as
\begin{align}
    D_s = \left\{ \ket{\psi} \in \qubits{n} : C_i \ket{\psi} = (-1)^{s_i} \ket{\psi} \ \forall i \in [N] \right\}.
\end{align}
In this language, $\Cc = D_{0^N}$. The logical operators $\Ll$ are the collection of Pauli operators that commute with every element of $\Ss$ but are not generated by $\Ss$:
\begin{align}
    \Ll = \left \{P \in \Pp_n : P C_i = C_i P \ \forall i \in [N] \right \} \setminus \Ss.
\end{align}
We say that the code is $\ell$-local if every $C_i$ is trivial on all but $\ell$ components of the tensor product and that each qubit of the code is non-trivial in at most $\ell$ of the checks $\{C_i\}$.
\end{definition}
Given a stabilizer code defined by $\{C_i\}_{i \in [N]}$, the associated local Hamiltonian is defined by
\begin{align}
    H = \sum_{i \in [N]} H_i \defeq \sum_{i \in [N]} \frac{\II - C_i}{2}.
\end{align}
This Hamiltonian is therefore commuting and furthermore it is a $\ell$-local low-density parity check Hamiltonian where $\ell$ is the locality of the code. Furthermore, the eigenspaces of $H$ are precisely the spaces $\{D_s\}$ with corresponding eigenvalues of $\abs{s}$, the Hamming weight of $s$. 
If the rate of the stabilizer code is $k$, we can identify a subset of $2k$ logical operators denoted as
\begin{align}
    \overline{X_1}, \overline{Z_1}, \ldots, \overline{X_k}, \overline{Z_k}
\end{align}
such that all operators square to identity and pairwise commute except $\overline{X_i}$ and $\overline{Z_i}$ which anti-commute for all $i \in [k]$.

\subsection{Circuits and lightcones}

We will assume that all quantum circuits in this work consist of gates with fan-in and fan-out of 2 and that the connectivity of the circuits is all-to-all. The gate set for the circuits will be the collection of all 2 qubit unitaries. Our results will be modified only by a constant factor if we assume gates with a larger constant bound on the fan-in and fan-out.

\begin{definition}[Circuit Complexity]
\label{def:cc}
Let $\rho$ be a mixed quantum state of $n$ qubits. Then the circuit complexity\footnote{We note that while our definition for circuit complexity of $\rho$ is given as the minimum depth of any circuit exactly generating a state $\rho$, we could have equivalently defined the circuit complexity of $\rho$ as the minimum depth of any circuit generating a state $\rho'$ within a small ball $B_\delta(\rho)$ of $\rho$ for some $\delta > 0$. This would not have changed our results except for constant factors. This is because our results will be concerned with lower-bounding the circuit complexity of all states of energy $\leq \eps n$. If $\rho$ is a state of energy $\leq \eps n$, then every state $\rho' \in B_\delta(\rho)$ has energy $\leq (\eps + \delta) n$. Therefore, by redefining $\eps \leftarrow \eps - \delta$, we can switch to the alternate definition of circuit complexity. We use the listed definition in our proofs as it vastly simplifies legibility.} of $\rho$, $\cd(\rho)$, is defined as the minimum depth over all $m$-qubit quantum circuits $U$ such that $U \ket{0}^{\otimes m} \in \qubits{m}$ is a purification of $\rho$. Equivalently,
\begin{align}
\cd(\rho) = \min \left\{ \depth(U) : \tr_{[m] \setminus [n]}\left(U \ketbra{0}{0}^{\otimes m} U^\dagger \right) = \rho \right\}.
\end{align}
\end{definition}

All product states including classical pure states have circuit complexity $0$ or $1$.
Given a state $\psi$ of circuit complexity $t$, each individual qubit of $\psi$ is entangled with at most $2^t$ other qubits, and therefore, states of constant circuit complexity are referred to as ``trivial'' or ``classical'' and likewise states of super-constant circuit complexity are inherently ``quantum'' and possess complex entanglement. Furthermore, properties of constant circuit complexity states are easy to classically verify: the energy $\tr(H \psi)$ of $\psi$ with respect to a local Hamiltonian $H$ can be computed classically in time $\exp(\exp(t))$.

Let $U = U_t \cdots U_1$ be a depth $t$ circuit acting on $\qubits{m}$, where each $U_j = \bigotimes_k u_{j,k}$ is a tensor product of disjoint two-qubit unitaries $u_{j,k}$. Fix a set of qubits $A\subset [m]$. We say that a qubit $i$ is in the lightcone of $A$ with respect to $U$ if the following holds. There is a sequence of successively overlapping two-qubit unitaries $\{u_{t,k_t}, u_{t-1,k_{t-1}},\ldots u_{b,k_b}\}$ (with $1\leq b \leq t$) such that the following holds: supports of $u_{t,k_t}$ and $A$ intersect, the supports of $u_{j,k_j}$ and $u_{j-1,k_{j-1}}$ intersect for all $b< j\leq t$, and the qubit $i$ is in the support of $u_{b,k_b}$. The support of the lightcone of $A$ with respect to $U$ is the set of qubits in the lightcone of $A$ with respect to $U$. We represent as $U_{A}$ the circuit obtained by removing all the two-qubit unitaries from $U$ not in the support of the light cone of $A$. We will use the following facts about lightcones.

\begin{fact}
\label{fact:lightfact2}
Consider a quantum state $\psi$ acting on $\qubits{m}$. For any $A\in [m]$, let $L_A$ denote the support of the lightcone of $A$ with respect to $U$. It holds that
\begin{align}\tr_{-A}(U\psi U^{\dagger}) = \tr_{-A}\left(U (\psi_{L_A}\otimes \nu_{-L_A})U^{\dagger}\right).
\end{align}
In other words, the reduced density matrix on qubit $A$, only depends on the reduced density matrix $\psi_{L_A}$ on the lightcone of $A$.
\end{fact}
\begin{proof}
For any operator $O$ supported on $A$, consider 
\begin{xalign}
\tr_A(O\tr_{-A}(U\psi U^{\dagger}))&=\tr(U^{\dagger}O U \psi) = \tr(U^{\dagger}O U \psi_{L_A}\otimes \nu_{-L_A})\\
&=\tr_A(O\tr_{-A}(U(\psi_{L_A}\otimes \nu_{-L_A})U^{\dagger})).
\end{xalign}
The second equality uses $U^{\dagger}O U=U_A^{\dagger}O U_A$ where $U_A$ is the circuit restricted to the region $A$. This proves the fact.
\end{proof}

\begin{fact}
\label{fact:lightfact3}
Consider a quantum state $\ket{\phi}=U\ket{0}^{\otimes m}$. Let $A\subset [m]$ and define $\ket{\phi'}=U_A\ket{0}^{\otimes m}$. We have $\tr_{-A}(\phi)=\tr_{-A}(\phi')$.
\end{fact}
\begin{proof}
The proof is very similar to that of Fact \ref{fact:lightfact2}. For any operator $O$ supported on $A$, consider 
\begin{align}
\tr(O\tr_{-A}(\phi))=\tr(U^{\dagger}O U \ketbra{0}{0}^m) = \tr(U_A^{\dagger}O U_A \ketbra{0}{0}^m)=\tr(O\tr_{-A}(\phi')).
\end{align}
This completes the proof.
\end{proof}

In our proofs, we will assume the simple upper bound of $2^t \abs{A}$ for the size of the lightcone generated by a depth $t$ circuit. This assumes all-to-all connectivity of the circuit. If the circuit was geometrically constrained to a lattice of a fixed constant dimension $D$, then the simple upper bound would be $O((tD )^D \abs{A})$. All our proofs can easily be translated into lower bounds for geometric circuits on a lattice using this substitution.

\subsection{Quantum PCP and NLTS conjectures}

We review the definitions of the quantum PCP and NLTS conjectures and their relationship.

\begin{conjecture}[Quantum PCP \cite{quant-ph/0210077,10.1145/2491533.2491549}]
It is $\QMA$-hard to decide whether a given $O(1)$-local Hamiltonian $H = H_1 + \cdots + H_m$ (where each $\norm{H_i} \leq 1$) has minimum eigenvalue at most $a$ or at least $b$ when $b - a \geq c \norm{H}$ for some universal constant $c > 0$.
\end{conjecture}

Since, it is widely believed that $\NP \neq \QMA$, solutions (i.e. ground-states) of $\QMA$-complete local Hamiltonians should not have classically checkable descriptions; for one, ground-states of $\QMA$-complete local Hamiltonians should not have constant circuit complexity. 

If we assume the quantum PCP conjecture, then \emph{all low-energy} states of the local Hamiltonian should not have constant circuit complexity. This is because if there exists a constant circuit complexity state of energy less than the promise gap, then the circuit description of the state can serve as a \emph{classically checkable witness} to the local Hamiltonian problem. This would place the promise gapped problem in $\NP$. Since the quantum PCP conjecture posits that promise gapped local Hamiltonians are also $\QMA$-complete, this would imply that $\NP \neq \QMA$, a contradiction. This is the inspiration for the NLTS conjecture.

\begin{conjecture}[NLTS \cite{10.5555/2600498.2600507}]
There exists a fixed constant $\eps > 0$ and an explicit family of $O(1)$-local Hamiltonians $\{H^{(n)}\}_{n=1}^\infty$, where $H^{(n)}$ acts on $n$ particles and consists of $\Theta(n)$ local terms, such that for any family of states $\{\psi_n\}$ satisfying 
\begin{align}\tr(H^{(n)} \psi_n) \leq \eps \norm*{H^{(n)}} + \lambda_{\min}(H^{(n)}),
\end{align}
the circuit complexity $\cd(\psi_n)$ grows faster than any constant\footnote{This definition is the one originally expressed by Freedman and Hastings in \cite{10.5555/2600498.2600507}. However, a consequence of the quantum PCP conjecture and $\NP \neq \QMA$ would be a circuit complexity lower bound of $\omega(\log \log n)$. For this reason, we will be more interested in circuit lower bounds of $\omega(\log \log n)$. Furthermore, if $\QCMA \neq \QMA$, then the necessary consequence of the quantum PCP conjecture is a circuit lower bound of $\omega(\poly(n))$. Our techniques make no obvious progress towards this strengthened conjecture as we study stabilizer codes whose circuit complexity is $O(\log n)$. Some progress towards super-polynomial NLETS was made by Nirkhe, Vazirani, and Yuen \cite{nirkhe_et_al:LIPIcs:2018:9095}.}.

\end{conjecture}

Therefore, making the widely believed assumption $\NP \neq \QMA$, the quantum PCP conjecture implies the NLTS conjecture, thereby identifying a necessary property of any $\QMA$-complete promise gapped local Hamiltonian. This property is referred to as ``robust entanglement'' since the entanglement complexity of every low-energy state must be non-trivial. Resolution of the NLTS conjecture is an important first step towards proving the quantum PCP conjecture. One of the advantages of the NLTS conjecture is that it does not involve complexity classes such as $\QMA$, but rather focuses solely on the entanglement complexity that is intrinsic to low-energy states of local Hamiltonians.

\section{Robust entropy-based lower bounds}
\label{sec:warmup-entropy}
In this section, we subsume a folklore proof for the circuit complexity of code-states to prove that the lower bound is robust to small trace-distance perturbations. %

\begin{lemma}
\label{lem:warmup-lemma}
Let $\Cc$ be a $[[n,k,d]]$ code and $\psi$ a state on $m$ qubits. Let $\psi_\code$ be the reduced state on the $n$ code qubits. If the trace-distance between $\psi_\code$ and $\Cc$ is $0 < \delta < 1/2$ and the code is of rate at least $k > 2 \delta \log(1/\delta) m$,
then the circuit complexity $cc(\psi) > \log d$.
\end{lemma}

\begin{proof}
Let $\psi$ be a state on $m$ qubits such that $\psi = U \ket{0}^{\otimes m}$ where $U$ is a circuit of depth $t$. Suppose $2^t<d$. Further assume that $\psi$ is $\delta$-close to the code $\Cc$ in trace distance, meaning that there exists a state $\rho_\code \in \Cc$ such that $\norm{\psi_{\code} - \rho_\code}_1 \leq \delta$. Thus, Uhlmann's theorem \cite{uhlmann76} ensures that there is a purification $\ket{\rho}$ on $m$ qubits such that $\norm{\ketbra{\psi}{\psi} - \ketbra{\rho}{\rho}}_1 \leq \delta$.

Let $\Enc$ be any encoding CPTP map from $\qubits{k} \rightarrow \qubits{n}$ mapping $k$ qubits to the $k$ qubit code-space. Define $\Ee$ as the maximally decohering channel as follows
\begin{align}
\Ee(\cdot) \defeq \frac{1}{4^k} \sum_{a,b \in \bits^k} \left( X^a Z^b \right) (\cdot) \left( X^a Z^b \right)^\dagger.
\end{align}
Then let $\Theta$ be the encoding of $\rho$ defined as
\begin{align}
\label{eq:maxmixcode}
    \Theta \defeq \Enc \circ \Ee \circ \inv{\Enc} (\rho).
\end{align}
This state is well-defined and has entropy $S(\Theta) \geq k$ since $S(\Ee(\rho)) \geq k$. We omit proof of this statement here as it is covered in greater generality by Fact \ref{fact:Eprop}.  
\begin{fact}[Extended local indistinguishability property]
\label{fact:extension-of-LI}
For any region $R_1 \cup R_2$ where $R_1$ is contained in the code qubits and $R_2$ in the ancilla qubits with $\abs{R_1} < d$, $\rho_{R_1 \cup R_2} = \Theta_{R_1 \cup R_2}$.
\end{fact}
We prove this fact after the lemma. Let $R \subset [m]$ be any region of the qubits of size $< d$.
Using this fact,
\begin{align}
    \norm{\psi_R - \Theta_R}_1 \leq \delta.
    \label{eq:distance-trace}
\end{align}
For any qubit $i \in [m]$, let $L_i \subset [m]$ be the support of the lightcone of $i$ with respect to $U$. The size of $L_i$, $\abs{L_i}$ is at most $2^t < d$. Applying Fact \ref{fact:lightfact2} here, we have
\begin{align}
    \label{eq:fact-trace}
    \tr_{-\{i\}}(U^\dagger \Theta U) = \tr_{-\{i\}}(U^\dagger (\Theta_{L_i}\otimes \nu_{-L_i}) U).
\end{align}
Since the size of $L_i$ is $< d$, we can combine \eqref{eq:distance-trace} and \eqref{eq:fact-trace} to achieve
\begin{align}
    \norm{\tr_{-\{i\}}(U^\dagger (\psi_{L_i}\otimes \nu_{-L_i}) U) - \tr_{-\{i\}}(U^\dagger (\Theta_{L_i}\otimes \nu_{-L_i}) U)}_1 \leq \delta.
\end{align}
However, $U^\dagger \psi U = \ketbra{0}{0}^{\otimes m}$ and so
\begin{align}
    \norm{\ketbra{0}{0} - \tr_{-\{i\}}(U^\dagger (\Theta_{L_i}\otimes \nu_{-L_i}) U)}_1 \leq \delta.
\end{align}
Using standard entropy bounds, we can bound the entropy of the $i$th qubit of the rotated state $\Theta$:
\begin{align}
    S \left( \tr_{-\{i\}}(U^\dagger \Theta U)  \right) = S \left( \tr_{-\{i\}}(U^\dagger (\Theta_{L_i}\otimes \nu_{-L_i}) U)  \right) \leq H_2(\delta) \leq 2 \delta \log(1/\delta).
\end{align}
Notice that $S(U^\dagger \Theta U)=S(\Theta) = k$. We can, therefore, bound $k$ by 
\begin{align}
    k \leq S(\Theta) \leq \sum_{i \in [m]} S \left( \tr_{-\{i\}}(U^\dagger \Theta U)  \right) \leq 2 \delta \log(1/\delta) m.
\end{align}
This leads to a contradiction since we assumed $k > 2\delta \log(1/\delta) m$.
\end{proof}

If $m = n$, then this provides a circuit lower bound for linear-rate codes. Fact \ref{fact:lightfact3} ensures that we can assume $m \leq 2^{\cc(\psi)} n$, without loss of generality. 
This gives us the following corollary:

\begin{corollary}
\label{cor:lin-rate-and-pure}
Let $\Cc$ be a $[[n,k,d]]$ code and $\ket \psi$ a pure-state and the trace-distance between $\ket{\psi}$ and $\Cc$ is $0 < \delta < 1/2$ such that $k > 2 \delta \log(1/\delta) n$. Then, the circuit complexity $\cc(\ket{\psi})$ satisfies
\begin{align}
    \cc(\ket{\psi}) \geq \log \left( \min \left\{ d, \frac{k}{2 \delta \log(1/\delta) n} \right\} \right).
\end{align}
\end{corollary}

\begin{proof}
By Lemma \ref{lem:warmup-lemma}, either $2^{\cc(\ket{\psi})} \geq d$ or $k \leq 2 \delta \log(1/\delta) 2^{\cc(\ket{\psi})} n$ since $m \leq 2^{\cc(\ket{\psi})} n$. Rearranging this is equivalent to the corollary.
\end{proof}

\begin{proof_of}{Fact \ref{fact:extension-of-LI}}
Let $R_1$ be a subset of the code qubits and $R_2$ be a subset of the ancilla qubits such that $\abs{R_1} < d$. We can express any code-state $\ket{\psi}$ over the $m$ qubits as
\begin{align}
    \ket{\psi} = \sum_{x \in \bits^k} \ket{\overline{x}} \ket{\psi_x}
\end{align}
where $\{\ket{\overline{x}}\}$ is a basis for the code and $\ket{\psi_x}$ are un-normalized. Let $U$ be any logical operator (i.e. one that preserves the code-space). Then,
\begin{align}
\tr_{-(R_1 \cup R_2)}\left(U \psi U^\dagger\right) = \sum_{x,y \in \bits^k} \tr_{-R_1}(U \ketbra{\overline{x}}{\overline{y}} U^\dagger) \otimes \tr_{-R_2} (\ketbra{\psi_x}{\psi_y})
\end{align}
In the summation, if $x = y$, then the first component is $\phi_{R_1}$ for some fixed state $\phi_{R_1}$ by local indistinguishability. Furthermore, if $x \neq y$, then the first component is $0$ by orthogonality of $U\ket{\overline{x}}$ and $U\ket{\overline{y}}$ despite the erasure of $\abs{R_1} < d$ qubits. Therefore,
\begin{align}
    \tr_{-(R_1 \cup R_2)} \left(U \psi U^\dagger \right) = \phi_{R_1} \otimes \sum_{x \in \bits^k}  \tr_{-R_2}(\psi_x).
\end{align}
which is an invariant of $U$, which means $\tr_{-(R_1 \cup R_2)}(\rho) = \tr_{-(R_1 \cup R_2)}(U \rho U^\dagger)$ . Since (a) $\Theta$ is a mixture over applications of \emph{logical} Paulis to $\rho$ and (b) a logical operator applied to a code-state is another code-state and therefore is locally indistinguishable, then it follows that $\rho_{R_1 \cup R_2} = \Theta_{R_1 \cup R_2}$.
\end{proof_of}

\subsection{Improved circuit lower bounds using AGSPs}

Corollary \ref{cor:lin-rate-and-pure} shows that if we are given a $[[n,k,d]]$ code with linear rate $k=\Omega(n)$, then a state generated by a depth $t \leq \gamma\log d$ circuit must be $\Omega\left(2^{-2t}\right)=\Omega\left(d^{-2\gamma}\right)$ far from the code-space in trace distance. Now we show an even stronger separation, if the code is the zero-eigenspace (ground-space) of a commuting local Hamiltonian. 

Consider a $[[n, k,d]]$ QLDPC code $\Cc$ which is the common zero-eigenspace of commuting checks $\{\Pi_j\}_{j=1}^N$ of locality $\ell$ each (this includes, but is not restricted to, the stabilizer code defined earlier). Consider a state $\ket{\psi}=U\ket{0}^{\otimes m}$ obtained by applying a depth $t$ circuit $U$ on $m$ qubits. Suppose there is a state $\rho_0\in \Cc$ having good fidelity with the code-space, that is, $f \defeq F(\psi_{\code},\rho_0)$. Fact \ref{fact:lightfact3} ensures that we can choose $m\leq 2^tn$. We prove the following lemma.  
\begin{lemma}
\label{lem:rootdwarmup}
For the state $\ket{\psi}$ as defined above, it holds that
\begin{align}2^{2t}\geq \min\left(d, \frac{1}{64\sqrt{\ell}\log^2 d\ell}\cdot\frac{k\sqrt{d}}{ n  \cdot \sqrt{\log\frac{1}{f}}}\right).\end{align}
\end{lemma}
The proof of this lemma appears in the Appendix \ref{sec:appendix-pf}. It uses the tool of approximate ground-space projectors (AGSP) and the principle that low min-entropy for gapped ground-states implies low entanglement entropy \cite{hastings2007area_law, AradLV12, AradKLV13, AHS20}. The lemma shows that if the code has linear rate $k=\Omega(n)$, then any state generated by a depth $t \leq \gamma\log d$ circuit must be $1-\exp\left(-\widetilde{\Omega}\left(d^{1-4\gamma}\right)\right)$ far from the code-space in trace distance. Here, the $\tilde{\Omega}$ notation hides some polylog factors. This bears some resemblance with the results of \cite{BeckIL12, lovett-viola}, which show that the distributions sampled from depth $t$ (and size $e^{n^{1/t}}$) classical $\AC0$ circuits are $1-e^{-n^{1/t}}$ far from the uniform distribution over a good classical code (linear rate and linear distance). A comparison with Lemma \ref{lem:rootdwarmup} is largely unclear, due to the differences between classical and quantum codes, as well as $\AC0$ circuits and quantum circuits.  

We further prove the following lemma for codes encoding at least one qubit and having large distance. This complements Lemma \ref{lem:rootdwarmup} replacing the linear rate condition with linear distance condition.

\begin{lemma}
\label{lem:distfidbound}
Given $\ket{\psi}=U\ket{0}^{\otimes m}$ with $U$ of depth $t$ and let $\Cc$ be a code with distance $d$. Let $\Pi$ be the projector onto the codespace and  $f=\sqrt{\bra{\psi}(\Pi\otimes \id_{\anc})\ket{\psi}}$ be the fidelity of $\ket{\psi}$ with the codespace. If $2^t\leq \frac{d}{2}$ then
\begin{align}f^2\leq 2e^{-\frac{d^2}{2^{2t+10}m}}.\end{align}
\end{lemma}
\begin{proof}
Let $G= \sum_{i=1}^m U\ketbra{1}_i U^{\dagger}$ be the $2^t$ local Hamiltonian with $\ket{\psi}$ as its unique ground state. From Fact \ref{fact:KLS}, there is a polynomial $P(G)$ of degree $\frac{d}{2^{t+1}}$ such that
\begin{equation}
\label{eq:KLSpoly}
\|P(G)-\ketbra{\psi}{\psi}\|\leq e^{-\br{\frac{d}{2^{t+1}}}^2/2^8m} = e^{-\frac{d^2}{2^{2t+10}m}}.    
\end{equation}
Note that each multinomial term in $P(G)$ is supported on $\leq 2^t\cdot \frac{d}{2^{t+1}} \leq \frac{d}{2}$ terms. Let $\ket{\phi}= \Pi\ket{\psi}/f$ be the codestate having largest overlap with $\ket{\psi}$. Consider any vector $\ket{\phi'}$ with $\phi'_{\code}\in \Pi$ that is orthogonal to $\ket{\phi}$ (and hence orthogonal to $\ket{\psi}$). One way to construct $\ket{\phi'}$ is to expand $\ket{\phi}=\sum_x \ket{\bar{x}}\ket{\phi_x}$ (with $\{\ket{\bar{x}}\}_x$ a basis for $\Pi$ and $\ket{\phi_x}$ unnormalized) and then define $\ket{\phi'}\propto\sum_x \alpha_x\ket{\bar{x}}\ket{\phi_x}$. The complex numbers $\{\alpha_x\}_x$ are chosen such that 
\begin{align}\sum_x \alpha_x \braket{\phi_x}{\phi_x}=0 \implies \bra{\phi}\ket{\phi'}=0.\end{align}
Since $P(G)$ is a sum of $\leq \frac{d}{2}$-local terms, \eqref{eq:kl-conditions} ensures that \begin{align}\bra{\phi'}P(G)\ket{\phi'}=\bra{\phi'}\Pi P(G)\Pi\ket{\phi'}=\eta_{P(G)}\bra{\phi'}\Pi \ket{\phi'}=\eta_{P(G)}=\bra{\phi}P(G)\ket{\phi}.
\end{align}
Using Equation \ref{eq:KLSpoly} and $\bra{\phi'}\ket{\psi}=0$, we have $\bra{\phi'}P(G)\ket{\phi'} \leq e^{-\frac{d^2}{2^{2t+10}m}}$. Thus using Equation \ref{eq:KLSpoly} again, 
\begin{equation}
\label{eq:fidupbd}
    f^2=\abs{\bra{\phi}\ket{\psi}}^2 \leq \bra{\phi}P(G)\ket{\phi} + e^{-\frac{d^2}{2^{2t+10}m}}=\bra{\phi'}P(G)\ket{\phi'} + e^{-\frac{d^2}{2^{2t+10}m}}\leq 2e^{-\frac{d^2}{2^{2t+10}m}}.
\end{equation}
This proves the lemma.
\end{proof}

Even with such a wide separation between the code-space and the states generated by low-depth circuits, these results give no insights into the energy such states can achieve. This is because these results do not rule out the possibility that such low-depth circuits live in the energy $1$ eigenspace of the code Hamiltonian, which is orthogonal to the code-space. In the later section, we show how to achieve this guarantee; this proves our main result.

\section{Lower bounds for stabilizer codes}
\label{sec:main}

In this section we prove Theorem \ref{thm:main}; {which combines Theorems \ref{thm:main-technical} and \ref{thm:distance_clb}. For Theorem \ref{thm:main-technical}}, we show how the entropy-based bounds from the previous section can be improved from handling states physically near the code-space to all low-energy states, once we assume that the code is a stabilizer code. The key property we exploit is that the local indistinguishability property of the code-space $\Cc$ also holds for each eigenspace $D_s$ in the case of stabilizer codes. We make this precise in the following facts; all facts are proven after the proof of the theorems.

The following fact argues that logical operators not only preserve the code-space $\Cc$ but rather any eigenspace $D_s$.

\begin{fact}
\label{fact:logact}
Fix a stabilizer code $\Cc$ on $n$ qubits with generator set $\{C_i\}_{i \in [N]}$. For any string $s \in \bits^{n}$, a state $\rho$ such that $\rho_\code \in D_s$, and a logical operator $P \in \Ll$, we have $(P\rho P)_{\code} \in D_s$.
\end{fact}

Each pair of Pauli operators either commute or anti-commute. The following fact imposes constraints on non-logical and non-stabilizer Pauli operators.

\begin{fact}
\label{fact:zeroexp}
Let $P$ be a Pauli operator such that for some $i \in [N]$, $P C_i = - C_i P$. For any $s \in \bits^{N}$ and any quantum state $\rho$ such that $\rho_\code \in D_s$, we have $\tr(P \rho) = 0$.
\end{fact}

The third crucial fact we will use is about the local indistinguishability of stabilizer codes. We will show that the measure of any local operator of locality $< d$, the distance of the code, is an invariant of each eigenspace $D_s$.

\begin{fact}
\label{fact:samemarginals}
Let $\rho$ be a state such that $\rho_{code}\in D_s$ for a string $s$. For a logical pauli $P\in \Ll$, define $\rho'=P\rho P$. It holds that for any region $T\subset [m]$ of size less than $d$, $\rho_T=\rho'_T$.
In general, let $\rho, \rho'$ be states such that $\rho_{\anc}=\rho'_{\anc}$ and $\rho_{code}, \rho'_{code}\in D_s$ for a string $s$. It holds that for any region $T\subset [m]$ of size less than $d$, $\rho_T=\rho'_T$.
\end{fact}

Since logical operators act like single qubit Pauli operators within the code-space, they can be used for randomization. %
Define the following quantum ``completely depolarizing in the logical basis'' channel that acts on code qubits, analogous to the channel defined in \eqref{eq:maxmixcode}:
\begin{align}
\label{eq:randomizer}
\cE(\cdot) \defeq \frac{1}{4^{k}} \sum_{a,b \in \bits^{k}} \left( \overline{X}^a \overline{Z}^b \right)(\cdot)\left( \overline{Z}^b\overline{X}^a  \right)
\end{align}
where $\overline{X}^a = \prod_i \overline{X_i}^{a_i}$ is a product of logical $X$ operators defined by $a$ and likewise $\overline{Z}^b$ is a product of logical $Z$ operators defined by $b$.
We will utilize the following two properties of this channel, analogous to Fact \ref{fact:extension-of-LI}.
\begin{fact}
\label{fact:Eprop}
It holds that
\begin{enumerate}
    \item For any quantum state $\rho$, the entropy $S(\cE(\rho))\geq k$.
    \item For any quantum state $\rho$ such that $\rho_{\code}\in D_s$ for some $s$, $\cE(\rho)_{\code}\in D_s$. Furthermore, for any set $T\subset [m]$ of size less than $d$, $\rho_T=\cE(\rho)_T$.
\end{enumerate}
\end{fact}

The next fact describes how all stabilizer terms of the code can be measured simultaneously using a short-depth circuit if the code has small locality. Let $N$ be the number of checks for an $\ell$-local code; recall that then $n/\ell \leq N \leq \ell n$.

\begin{fact}
\label{fact:efficient_measurement}
Let $\Cc$ be a stabilizer code of locality $\ell$ on $n$ qubits with $N$ checks $\{C_i\}_{i \in [N]}$. Then, there is a circuit $V$ of depth $\leq 2 \ell^3$ which coherently measures the value of each stabilizer term into $N$ ancilla.
\end{fact}

Lastly, consider a state $\ket{\phi}=U\ket{0}^{\otimes m}$ where $U$ is a circuit of depth $t$. From Fact \ref{fact:lightfact3}, we can assume $m\leq n 2^t$ without loss of generality. We are now ready to state and prove {the following theorem for codes of large rate.} %

\begin{theorem}
\label{thm:main-technical}
Let $\Cc$ be a $[[n,k,d]]$ stabilizer code of locality $\ell$ defined by checks $\{C_i\}_{i \in [N]}$. Let $H$ be the corresponding Hamiltonian. Suppose there is a state $\ket{\phi}$ on $m$ qubits with $\tr(H\phi) \leq \eps N$ and circuit complexity $t \defeq \cd(\phi) < \log(d) - 2 \ell^3$. Then, for a constant $c_\ell$ depending only on $\ell$ and not the size of the code,
\begin{align}
    2^{2t} > \frac{k}{c_\ell n \cdot \eps\log\frac{1}{\eps}}.
\end{align}
\end{theorem}

\begin{proof}
All stated intermediate claims are proven in the next sub-section. By assumption, $\ket{\phi} = U \ket{0}^{\otimes m}$ for a circuit of depth $t < \log(d) - 2 \ell^3$. Define the energy of each local Hamiltonian term $H_i$ as
\begin{align}
    \eps_i \defeq \tr(H_i \phi) = \frac{1}{2} - \frac{1}{2} \tr(C_i \phi).
\end{align}
Add $N \leq n$ new syndrome-measurement ancilla (SMA) qubits each with initial state $\ket{0}$ and coherently measure the entire syndrome using the depth $2 \ell^3$ circuit $V$ from Fact \ref{fact:efficient_measurement}. Then the state
\begin{align}
    \ket{\psi} = V \left(\ket{\phi} \otimes \ket{0}^{\otimes N}\right) = V U \ket{0}^{\otimes(m + N)} \defeq W \ket{0}^{\otimes (m + N)}
\end{align}
with $W = VU$ a circuit of minimum circuit depth $\defeq \cd(W) \leq t + 2 \ell^3$.
Define the state obtained by incoherently measuring all the SMA qubits of $\ket{\psi}$ as 
\begin{align}
    \Psi = \sum_{s \in \bits^{N}} D_s \ketbra{\phi}{\phi} D_s \otimes \ketbra{s}{s}
\end{align}
where we abuse notation slightly and use $D_s$ both as the eigenspace and the projector onto it. Since we assume that $\Cc$ is a stabilizer code, the Hamiltonian terms $H_i$ all mutually commute and therefore so do the measurements of the SMA qubits. Therefore, the order of measurement used is irrelevant.

Define the state $\Theta = \cE(\Psi)$ obtained by applying the logical completely depolarizing channel $\cE$ from \eqref{eq:randomizer}. Then, we have
\begin{align}
    \Theta = \sum_s \tr(D_s \phi D_s) \mu_s \otimes \ketbra{s}{s} \qquad \text{ for } \mu_s \defeq \cE \left( \frac{D_s \phi D_s}{\tr(D_s \phi D_s)} \right).
\end{align}

\begin{claim}
\label{clm:gentlemeas}
Fix any region $R\subset [m+N]$. Let $S_R$ be the set of all indices $i\in [N]$ such that the $i$th SMA qubit belongs to $R$. It holds that 
\begin{align}
\label{eq:closelowen}
F(\psi_R,\Psi_R)\geq 1-\sum_{i\in S_R}\eps_i.
\end{align}
Further, if $\abs{R}<d$, then $\Psi_R=\Theta_R$.
\end{claim}
For every $j\in [m+N]$, let $L_j$ be the support of the lightcone of $j$ with respect to the unitary $W^{\dagger}$. Note that $|L_j|\leq 2^{\cd(W)} < d$. Since $W^{\dagger}\ket{\psi}$ is $\ket{0}^{\otimes(m+N)}$, we have that for any qubit $j\in [m+N]$,
\begin{align}
\label{eq:purezeroent}
    \tr_{-\{j\}} \left( W^\dagger \psi W \right) = \ketbra{0}{0}.
\end{align}
However, Fact \ref{fact:lightfact2} allows us to equate
\begin{xalign}
\tr_{- \{j\}} \left( W^\dagger \psi W \right) &= \tr_{- \{j\}} \left( W^\dagger (\psi_{L_j}\otimes \nu_{-L_j}) W \right), \\
\tr_{- \{j\}} \left( W^\dagger \Theta W \right) &= \tr_{- \{j\}} \left( W^\dagger (\Theta_{L_j}\otimes \nu_{-L_j}) W \right).
\end{xalign}
Using \eqref{eq:closelowen}, we find that for all $j \in [m+N]$,
\begin{xalign}
& F\br{\tr_{-\{j\}} \left( W^\dagger \psi W \right), \tr_{-\{j\}} \left(W^\dagger \Theta W \right)} \\
&\quad = F\br{ \tr_{-\{j\}} \left( W^\dagger (\psi_{L_j}\otimes \nu_{-L_j}) W \right), \tr_{-\{j\}} \left(W^\dagger (\Theta_{L_j}\otimes \nu_{-L_j}) W \right)} \\
&\quad \geq F\br{\psi_{L_j},\Theta_{L_j}} \\
&\quad \geq 1-\sum_{i\in S_{L_j}}\eps_i. \label{eq:psicloseto0}
\end{xalign}
We now infer from \eqref{eq:purezeroent} and \eqref{eq:psicloseto0} that\footnote{Given a binary distribution $(p,1-p)$, we can upper bound its entropy as follows. If $p\geq \frac{1}{4}$, then an upper bound is $1$. Else the upper bound is $2p\log\frac{1}{p}$. The combined upper bound is $2p\log\frac{1}{\min(p,\frac{1}{4})}$.}
\begin{align}
    S \left( \tr_{-\{j\}} \left(W^\dagger \Theta W\right) \right) \leq  2\left(\sum_{i\in S_{L_j}}\eps_i\right)\log \frac{1}{\min\left(\sum_{i\in S_{L_j}}\eps_i, \frac{1}{4}\right)}.
\end{align}
Using the concavity of the function $x \mapsto x\log\frac{1}{\min\left(x,\frac{1}{4}\right)}$ in the interval $x\in (0, 2^{\cd(W)})$, we can average over all $j\in [m+N]$ to conclude
\begin{xalign}
    \Exp_{j\in [m+N]}&S \left( \tr_{-\{j\}} \left(W^\dagger \Theta W\right) \right) \\
    &\leq  2\Exp_{j\in [m+N]}\left(\left(\sum_{i\in S_{L_j}}\eps_i\right)\log \frac{1}{\min\left(\sum_{i\in S_{L_j}}\eps_i, \frac{1}{4}\right)}\right)\label{eq:totent} \\
    &\leq 2\cdot \left(\Exp_{j\in [m+N]}\sum_{i\in S_{L_j}}\eps_i\right)\log \frac{1}{\min\left(\Exp_{j\in [m+N]}\sum_{i\in S_{L_j}}\eps_i, \frac{1}{4}\right)}. 
\end{xalign}
The next claim helps upper and lower bound this expression.
\begin{claim}
\label{clm:mostlow}
It holds that 
\begin{align}
\frac{\eps N}{m+N}\leq \Exp_{j\in [m+N]}\sum_{i\in S_{L_j}}\eps_i\leq 2^{2\cd(W)}\frac{\eps N}{m+N}.
\end{align}
\end{claim}
We now upper bound the entropy of $\Theta$. For this, let us assume $2^{\cd(W)}\leq \frac{1}{\eps}$, else the proof is immediate. 
\begin{xalign}
S(\Theta) &= S(W^\dagger \Theta W)\leq \sum_{j \in [m+N]} S\left( \tr_{-\{j\}}\left(W^\dagger \Theta W \right) \right) \label{eq:entropysum-1} \\
&\leq 2^{1+2\cd(W)}\eps N \log \frac{1}{\min\left(\frac{\eps N}{m+N}, \frac{1}{4}\right)} \label{eq:entropysum-2} \\
&\leq
2^{1+2\cd(W)}\eps \ell n \log \frac{2^{\cd(W)}}{\eps} \label{eq:entropysum-3} \\
&\leq \left(2^{2+2\cd(W)}\ell n \right) \cdot \eps \log \frac{1}{\eps}. \label{eq:entropysum-4}
\end{xalign}
The inequality in \eqref{eq:entropysum-1} comes from the subadditivity of entropy; the inequality in \eqref{eq:entropysum-2} uses \eqref{eq:totent} and then substitutes the upper and lower bounds given in Claim \ref{clm:mostlow}; the inequality in \eqref{eq:entropysum-3} uses $ \frac{n}{\ell}\leq N\leq \ell n$ and $\frac{\eps N}{m+N}\geq \frac{\eps}{ \ell 2^{t}+1}\geq  \frac{\eps}{ 2^{\cd(W)}}$; the inequality in \eqref{eq:entropysum-4} uses $2^{\cd(W)}\leq \frac{1}{\eps}$. Furthermore, $\Theta$ is the output of $\cE$ acting on $\Psi$. By Fact \ref{fact:Eprop} (Item 2), $S(\Theta) \geq k$. Combining the lower and upper bounds on the entropy of $\Theta$, the proof concludes. 
\end{proof}

Next, we prove the following theorem for codes of large distance.
\begin{theorem}
\label{thm:distance_clb}
Let $\Cc$ be a $[[n,k,d]]$ stabilizer code of locality $\ell$ defined by checks $\{C_i\}_{i \in [N]}$. Let $H$ be the corresponding Hamiltonian. Suppose there is a state $\ket{\phi}$ on $m$ qubits with $\tr(H\phi) \leq \eps N$ and circuit complexity $t \defeq \cd(\phi) < \log(d) - 1$. Then, 
\begin{align}
2^{2t} \geq \frac{d}{2^{6}n\sqrt{\ell\eps \log\frac{1}{\eps}}}.
\end{align}

\end{theorem}
\begin{proof}
Since $\tr(H\phi) \leq \eps N$, Markov's inequality ensures that $\tr(D_{\leq 2\eps N}\phi) \geq \frac{1}{2}$, where $D_{\leq 2\eps N}$ is the subspace of energy $\leq 2\eps N$. Since $D_{\leq 2\eps N} = \sum_{s\in \{0,1\}^N: |s|\leq 2\eps N} D_s$ and the number of $s$ satisfying $|s|\leq 2\eps N$ is $\leq 2^{4\eps N \log \frac{1}{\eps}}$, there exists a $s^*\in \{0,1\}^N$ such that $\tr(D_{s^*}\phi) \geq 2^{-4\eps N \log \frac{1}{\eps}-1}$. Now, Fact \ref{fact:samemarginals} ensures that $D_{s^*}$ is also an error correcting code of distance $d$. Applying Lemma \ref{lem:distfidbound} (assuming $2^t\leq \frac{d}{2}$) and setting $m\leq 2^t n$ (using Fact \ref{fact:lightfact3}), we conclude
\begin{xalign}
2^{-4N \eps \log \frac{1}{\eps}-1} &\leq  2e^{-\frac{d^2}{2^{3t+10}n}} \\
&\implies 2^{4N \eps \log \frac{1}{\eps}+2} \geq  e^{\frac{d^2}{2^{3t+10}n}} \\
&\implies 2^{3t}\geq \frac{d^2}{2^{12}nN\eps \log\frac{1}{\eps}}.
\end{xalign}
Since $N\leq n\ell$ and $2^{3t}\leq 2^{4t}$, the proof concludes.
\end{proof}
Combining Theorems \ref{thm:main-technical} and \ref{thm:distance_clb}, the proof of Theorem \ref{thm:main} concludes.

\subsection{Omitted proofs}

\begin{proof_of}{Fact \ref{fact:logact}}
For all $i \in [N]$ it suffices to verify that
\begin{align}
C_i (P \rho P) = P C_i \rho P = P (-1)^{s_i} \rho P = (-1)^{s_i} (P \rho P).
\end{align}
\end{proof_of}

\begin{proof_of}{Fact \ref{fact:zeroexp}}
\begin{align}
    \tr(P \rho) = \tr \left( P \left( (-1)^{s_i} C_i \rho \right) \right) = (-1)^{s_i + 1} \tr(C_i P \rho) = (-1)^{s_i + 1} \tr(\rho C_iP) = - \tr(P \rho).
\end{align}
where we used the cyclicality of trace twice. 
\end{proof_of}

\begin{proof_of}{Fact \ref{fact:samemarginals}}
Consider the first part of the fact. If the region $T$ lies entirely on the ancilla qubits, then the claim is easy since $P$ is trivial on the ancilla qubits.  Thus, consider a region $T=T_{\code} \cup T_{\anc}$ with code region $T_{\code}$ and ancilla region $T_{\anc}$. From Fact \ref{fact:logact}, $\rho'_{\code}\in D_s$. Suppose $\rho'_T\neq \rho_T$. Then there is a Pauli operator $P'$ of weight $\leq \abs{T}$ distinguishing the two states: 
\begin{align}
\tr(P' \rho_T)\neq \tr(P' \rho'_T).
\end{align}
This can be re-written as 
\begin{align}
\tr(P'\rho) \neq \tr(P' \rho') =  \tr(PP' P\rho).    
\end{align}
This relation holds only if 
\begin{enumerate}
    \item $\Tr(P' \rho)\neq 0$ or $\Tr(P' \rho')\neq 0$,  and 
    \item $PP'=-P'P$.
\end{enumerate}
The first relation ensures that $P'$ commutes with all $C_i$ (Fact \ref{fact:zeroexp}) and hence $P'_{\code}$ commutes with all $C_i$. Further, $P'_{\code}\notin \Ss$, else we would have $PP'_{\code}=P'_{\code}P$ which would imply $\Tr(PP' P\rho)=\Tr(P'\rho)$. Thus $P'_{\code}$ is a logical operator of weight $\leq |T|< d$. This suffices to establish a contradiction and prove the first part. But we can go further: the second relation implies that $P'_{\code}$ anti-commutes with $P$. But this is also a contradiction if $|T|$ is less than the weight of the smallest logical Pauli  anti-commuting with $P$. This will be useful in Section \ref{sec:other-techniques}. 

The second part follows similarly. Consider a region $T=T_{\code}\cup T_{\anc}$ such that $\rho_T \neq \rho'_T$. Then there is a Pauli $P'$ of weight $< d$ such that $\Tr(P'\rho)\neq \Tr(P'\rho')$. Clearly, $P'_{\code}$ must be non-identity and one of $\Tr(P\rho)$ or $\Tr(P\rho')$ should be non-zero. Due to Fact \ref{fact:zeroexp}, $P$ must commute with all $C_i$. This implies that $P_{\code}$ must also commute with all $C_i$. But then $P_{\code}$ is a logical operator with weight less than $d$, a contradiction.
\end{proof_of}

\begin{proof_of}{Fact \ref{fact:Eprop}}
For the first item, note that there exists an isometry $V$ such that
\begin{align}
V\overline{X_j} V^{\dagger} = X_j, \quad V\overline{Z_j} V^{\dagger} = Z_j, \forall j\in [k],
\end{align}
where $\{(X_j,Z_j)\}_{j\in [k]}$ are pairs of single qubit pauli operators on $k$ qubits $Q_1,\ldots Q_k$. The map 
$V\cE(V^{\dagger}(.)V)V^{\dagger}$ is the completely depolarizing map on these $k$ qubits. Since the completely depolarizing map transforms any quantum state $\sigma$ to $\tr_{Q_1,\ldots Q_k}(\sigma)\otimes \frac{\id_{Q_,\ldots Q_k}}{2^k}$, the statement follows.

The first part of the second item is a consequence of Fact \ref{fact:logact}. For the second part of the second item, we use an argument similar to Fact \ref{fact:samemarginals}. Since $\cE$ only acts on code qubits, $T$ must have support on code qubits. Suppose $\rho_T\neq \cE(\rho)_T$ and let $P$ be a Pauli such that $\Tr(P\rho)\neq \Tr(P\cE(\rho))$. Since one of the two terms is non-zero, Fact \ref{fact:zeroexp} ensures that $P$ commutes with all the $C_i$. But $P_{\code}\notin \cS$, else $\cE(P)=P$. Thus, $P_{\code}$ is a logical Pauli operator. This implies that the weight of $P_{\code}$ must be at least $d$, a contradiction. 
\end{proof_of}

\begin{proof_of}{Fact \ref{fact:efficient_measurement}}
Since each code qubit acts non-trivially in at most $\ell$ checks and each check $C_i$ is of size $\ell$, then each check $C_i$ overlaps non-trivially with at most $\ell^2$ other checks. Consider a graph defined by vertices $i \in [N]$ and edges whenever checks overlap non-trivially; this graph has degree $\ell^2$ and is, therefore, $\ell^2 + 1$-colorable. 

The following unitary $V_i$ coherently measures the stabilizer check $C_i$:
\begin{align}
    V_i \ket{\omega} \ket{y} \defeq \frac{\II + C_i}{2} \ket{\omega}\ket{y} + \frac{\II - C_i}{2} \ket{\omega}\ket{y \oplus 1}.
\end{align}
Furthermore, there is a depth $\ell$ circuit which calculates $V_i$. By the coloring argument, we can produce a depth $\ell(\ell^2 + 1)$ circuit $V$ that coherently measures all the stabilizers; namely, we apply sequentially all the unitaries $V_i$ per color.
\end{proof_of}

\begin{proof_of}{Claim \ref{clm:gentlemeas}}
Let $\psi'$ be the state obtained from $\psi$ by measuring all the SMA qubits in $S_R$. Notice that $\psi'_R=\Psi_R$. Thus, by appealing to data-processing, we have
\begin{align}
F(\Psi_R,\psi_R)=F(\psi'_R,\psi_R)\geq F(\psi',\psi).    
\end{align}
Next, we lower bound $F(\psi',\psi)$. We can represent $\ket \psi$ by
\begin{align}
\ket{\psi}=\sum_{t\in \{0,1\}^{|S_R|}}\sqrt{P(t)}\ket{\psi_t}\ket{t}_{S_R},\quad \psi'=\sum_{t\in \{0,1\}^{|S_R|}}P(t)\ketbra{\psi_t}\otimes\ketbra{t}_{S_R},
\end{align}
where $\ket{\psi_t}$ are normalized states. We have
\begin{xalign}
F^2(\psi,\psi') &= \sum_{t\in \{0,1\}^{|S_R|}}\sqrt{P(t)}P(t)\sqrt{P(t)}\abs{\braket{\psi_t}{\psi_t}}^2 \\ &=\sum_{t\in \{0,1\}^{|S_R|}}P(t)^2 \\
&\geq P(t = (0,\ldots,0))^2.
\end{xalign}
Using a union bound, we get 
\begin{align}
P(t = (0,\ldots,0)) \geq 1- \sum_{i=1}^{\abs{S_R}}P(t_i=1)= 1-\sum_{i\in S_R}\eps_i,
\end{align}
where the last equality uses the definition of the energy $\eps_i$. Thus,
\begin{align}
F(\psi,\psi')= P(t = (0,\ldots,0)) \geq 1-\sum_{i\in S_R}\eps_i.
\end{align}
This completes the first part of the Claim. For the second part, we consider for every $s\in \bits^N$:
\begin{xalign}
(\mu_s\otimes \ketbra{s})_{R}&=(\mu_s)_{R\setminus S_R}\otimes (\ketbra{s})_{S_R}\\
&=\left(\frac{D_s\phi D_s}{\Tr(D_s\phi D_s)}\right)_{R\setminus S_R}\otimes (\ketbra{s})_{S_R}\\
&=\left(\frac{D_s\phi D_s}{\Tr(D_s\phi D_s)}\otimes \ketbra{s}\right)_{R}.
\end{xalign}
The second equality uses Fact \ref{fact:Eprop} (Item 2) as $|R\setminus S_R|\leq |R|<d$. This establishes $\Psi_R=\Theta_R$.
\end{proof_of}

\begin{proof_of}{Claim \ref{clm:mostlow}}
Consider
\begin{align}(m+N)\Exp_{j\in [m+N]}\sum_{i\in S_{L_j}}\eps_i=\sum_{j\in [m+N]}\sum_{i\in S_{L_j}}\eps_i=\sum_{i}\eps_i|\{j:i\in S_{L_j}\}|.\end{align}
Let us upper bound $|\{j:i\in S_{L_j}\}|$ for any given $i$. Suppose $j\in [m+N]$ is such that $S_{L_j}$ contains a particular SMA qubit $i$. Then $i$ lies in the support of the lightcone of $j$ with respect to $W^{\dagger}$. This puts a constraint on the set of possible $j$'s: for some $b \in [\cd(W)]$, $j$ must lie in the lightcone of $i$ with respect to the circuit $W^{(b)}$ defined as the last $b$ layers of $W$. The number of $j$'s which satisfy this, for a given $i$, is at most $\sum_{b=1}^{\cd(W)}2^{b}\leq 2^{2\cd(W)}$. Thus, 
\begin{align}\sum_{j\in [m+N]}\sum_{i\in S_{L_j}}\eps_i\leq 2^{2\cd(W)}\sum_{i}\eps_i=2^{2\cd(W)}\eps N.\end{align}
The lower bound follows since $|\{j:i\in S_{L_j}\}|\geq 1$. This completes the proof.
\end{proof_of}

\section*{Acknowledgments}
Both authors were supported by NSF Quantum Leap Challenges Institute Grant number OMA2016245 and part of this work was completed while both authors were participants in the Simons Institute for the Theory of Computing program on \emph{The Quantum Wave in Computing}. Chinmay Nirkhe also acknowledges support from an IBM Quantum PhD Internship.
We thank Lior Eldar, Zeph Landau, Umesh Vazirani, and an anonymous conference reviewers for detailed comments on the manuscript that greatly improved the presentation. Additional thanks to Dorit Aharanov, Matt Hastings, Aleksander Kubica, Rishabh Pipada, Elizabeth Yang, and Henry Yuen for helpful discussions. 

\appendix

\section{Additional techniques for circuit lower bounds}
\label{sec:other-techniques}

In this appendix, we include other techniques for lower bounds on the circuit depth of error-correcting codes. We developed these techniques along the way, but are not necessary for our main result. However, they offer related, and yet different, techniques for circuit lower bounds and may be of independent interest.

\subsection{Robust circuit lower bounds for linear-distance codes}

Similar to the proof of the circuit lower bound for all states physically close to the code-space in the case of linear-rate codes, we provide a similar proof in the case of linear-distance codes. The intuition is the same: show that the distance of the $\ket{0}^{\otimes n}$ state from any code-space is dependent on $d/n$ and argue its consequence for the original code. The bounds apply for ancilla-free circuits.

\begin{lemma}
\label{lem:lineardist}
Let $\Cc$ be a $[[n,k,d]]$ code for $k > 0$. Then the ancilla-free circuit complexity of any state $\sigma$ on $n$ qubits with trace distance $\delta$ from $\Cc$ is at least $\Omega(\log (\frac{d}{\delta n}))$.
\end{lemma}

\begin{proof}
Let $U$ be a depth $t$ circuit generating the state $U\ket{0}^{\otimes n}$ and suppose it has distance $\delta$ from the $[[n,k,d]]$ code $\Cc$ (i.e. $\norm{\rho - \ketbra{0}{0}^{\otimes n}}_1 \leq \delta$). Then $\ketbra{0}{0}^{\otimes n}$ has distance $\delta$ from the code $U^\dagger\Cc U$ which is a $[[n,k,d/2^t]]$ code. Lemma \ref{lem:rotatedcodetrdist} shows that $\delta = \Omega \left(\frac{d}{2^t n} \right)$, completing the proof.
\end{proof}

\begin{lemma}
\label{lem:rotatedcodetrdist}
Let $k>0$ and consider a $[[n,k,d]]$ code. The trace distance between $\ket{0}^{\otimes n}$ and the code-space is at least $\Omega(d/n)$.
\end{lemma}

\begin{proof}
Assume that the distance between $\ketbra{0}{0}^{\otimes n}$ and some state $\rho$ of the $[[n,k,d]]$ code is $\delta \in (0,1)$. Define $\Pi_R =  \prod_{i \in R} \ketbra{0}_i$ for any $R\subset [n]$. Divide $[n]$ into $\frac{2n}{d}$ disjoint sets $\{R_j\}_{j=1, \ldots, \frac{2n}{d}}$, each of size at most $d - 1$. For any $\rho'$ in the code-space, local indistinguishability implies $\forall \ i$,
\begin{align}
\Tr(\Pi_{R_i}\rho')=\Tr(\Pi_{R_i}\rho) \geq \Tr(\Pi_{R_i}\ketbra{0}{0}^{\otimes n}) -\delta = 1-\delta.
\end{align}
Using the `union bound' inequality
\begin{align}\id-\ketbra{0}{0}^{\otimes n} \preceq \sum_{i=1}^{\frac{2n}{d}}(\id-\Pi_{R_i}),\end{align}
we thus find
\begin{align}\Tr((\id-\ketbra{0}{0}^{\otimes n})\rho')\leq \sum_{i=1}^{\frac{2n}{d}}\Tr((\id-\Pi_{R_i})\rho') \leq \frac{2n\delta}{d}.\end{align}
If $\delta \leq \frac{d}{6n}$, this implies that all code-states $\rho'$ have fidelity at least $\frac{2}{3}$ with $\ketbra{0}{0}^{\otimes n}$. But this leads to a contradiction as one can choose two orthogonal code-states; the dimension of the code being at least $2$. This completes the proof.
\end{proof}

\subsection{Best distance code lower bounds}

Attempts at the CNLTS conjecture may require circuit lower bound methods that are robust to the removal of a small fraction of code Hamiltonian checks. Unfortunately, circuit lower bounds in terms of distance do not appear to be robust: removing checks from a stabilizer code may significantly reduce the code distance. But, it is plausible that the distance associated with some pairs of logical operators does not decrease. We call this the `best distance' of the code (formalized below and implicit in \cite{8104078}) and prove lower bounds in terms of this notion. Our lower bound is inspired by the use of the uncertainty principle from \cite{8104078}, but provides a statement that may be incomparable to theirs.

Fix a stabilizer code and consider a logical Pauli pair $\overline{X_i},\overline{Z_i}$. We drop the subscript $i$ and write $\overline{X},\overline{Z}$.
Let $w$ be the maximum of two weights, $\abs{\overline{X}}$ and $\abs{\overline{Z}}$. 
Let $d'$ be the weight of the smallest logical operator that anti-commutes with either of $\overline{X}$ or $\overline{Z}$. The best distance of the code is the largest value of $d'$ over all logical pairs. Note that $d'\geq d$. Eldar and Harrow \cite[Section 4: Lemma 37]{8104078} show that for any state $\ket{\psi}$, 
\begin{equation}
\label{eq:uncertainty}
\bra{\psi}\overline{X}\ket{\psi}^2 + \bra{\psi}\overline{Z}\ket{\psi}^2 \leq 1.
\end{equation} 
This is interpreted as an uncertainty principle. Now we show the following.
\begin{lemma}
Consider a product state $\ket{\psi}= \bigotimes_{j = 1}^n \ket*{u_j}$. Then for any code-state $\ket{\rho}$, we have $\frac{1}{2}\|\psi-\rho\|_1\geq \frac{d'}{8w}$.
\end{lemma}
\begin{proof}
We assume for contradiction that there is a code-state $\ket{\rho}$ with $\frac{1}{2}\|\psi-\rho\|_1< \frac{d'}{8w}$.  From \ref{eq:uncertainty}, we have one of the following possibilities:
\begin{align}
\abs{\bra{\psi}\overline{X}\ket{\psi}}\leq \frac{1}{\sqrt{2}}, \quad \abs{\bra{\psi}\overline{Z}\ket{\psi}}\leq \frac{1}{\sqrt{2}}.    
\end{align}
Without loss of generality, assume the first holds.  Define the product state $\ket{\theta}=\overline{X}\ket{\psi}$. Thus, $F(\psi, \theta)\leq \frac{1}{\sqrt{2}}$. Let $L$ be the set of qubits supporting $\overline{X}$. Divide $L$ into distinct parts $L_1,L_2, \ldots L_{w/d'}$ such each part has size at most $d'-1$. Since $\psi$ is a product state, 
\begin{align}F(\psi_{L_1}, \theta_{L_1})F(\psi_{L_2}, \theta_{L_2})\ldots F(\psi_{L_{w/d'}}, \theta_{L_{w/d'}})=F(\psi, \theta)\leq \frac{1}{\sqrt{2}}.
\end{align} 
Thus, for at least one of $L_j$ (take $L_1$ without loss of generality), we have
\begin{align}F(\psi_{L_1}, \theta_{L_1})\leq \frac{1}{2^{d'/2w}}\implies \frac{1}{2}\|\psi_{L_1}- \theta_{L_1}\|_1 \geq \frac{d'}{4w}.
\end{align}
On the other hand, the proof assumes $\frac{1}{2}\|\psi-\rho\|_1< \frac{d'}{8w}$ which implies $\frac{1}{2}\|\theta-\overline{X}\rho \overline{X}\|_1< \frac{d'}{8w}$. Since $\rho$ is a code-state and $L_1$ has size at most $d'-1$, first part of the Fact \ref{fact:samemarginals} (i.e. local indistinguishability) ensures that $\rho_{L_1}=(\overline{X}\rho \overline{X})_{L_1}$ (as noted in its proof, the first part of this fact applies with the distance $d$ replaced by $d'$). This implies via triangle inequality that $\frac{1}{2}\|\psi_{L_1}-\theta_{L_1}\|_1< \frac{d'}{4w}$, which is a contradiction.
\end{proof}
Product states have circuit complexity $\leq 1$. We extend this argument to the case of low-depth circuits. 
\begin{lemma}
\label{lem:lowdepthunc}
Consider a state $\ket{\psi}= U\ket{0}^{\otimes m}$ on $m$ qubits, where $U$ has depth $t$ and $t\leq \log\frac{d'}{2}$. Then for any state $\rho_{\code}\in D_{s}$ for some $s\in \{0,1\}^N$, we have \begin{align}\frac{1}{2}\|\psi_{\code}-\rho_{\code}\|_1\geq \frac{1}{2}\br{\frac{d'}{2^{2t+6}w}}^2.
\end{align}
\end{lemma}
The lemma uses the following well known fact.
\begin{fact}
\label{fact:uhlmann}
For any quantum state $\ket{\psi}$ on $m$ qubits and a quantum state $\rho_{\code}$ on code qubits, there is a quantum state $\ket{\rho}$ on $m$ qubits such that
\begin{align}\frac{1}{2}\|\psi_{\code}-\rho_{\code}\|_1\geq \frac{1}{8}\|\psi-\rho\|^2_1.\end{align}
\end{fact}
\begin{proof}
By Uhlmann's theorem \cite{uhlmann76}, there is a quantum state $\ket{\rho}$ purifying $\rho_{\code}$ on $m-n$ qubits such that
\begin{align}F(\psi_{\code}, \rho_{\code})=F(\psi,\rho).\end{align}
Thus,
\begin{xalign}
\frac{1}{2}\|\psi_{\code}-\rho_{\code}\|_1 &\geq 1-\sqrt{F(\psi_{\code},\rho_{\code})} \\
&=1-\sqrt{F(\psi,\rho)} \\
&\geq 1-\sqrt{1-\frac{1}{4}\|\psi-\rho\|^2_1} \\
&\geq \frac{1}{8}\|\psi-\rho\|^2_1.
\end{xalign}
\end{proof}

\begin{proof_of}{Lemma \ref{lem:lowdepthunc}}
We assume for contradiction that there is a $\rho_{\code}\in D_{s}$ (for some $s$) such that $\frac{1}{2}\|\psi_{\code}-\rho_{\code}\|_1< \frac{1}{2}\br{\frac{d'}{2^{2t+6}w}}^2$. Fact \ref{fact:uhlmann} ensures that there is a purification $\ket{\rho}$ of $\rho_{\code}$ on $m$ qubits such that $\frac{1}{2}\|\psi-\rho\|_1< \frac{d'}{2^{2t+6}w}$.  From \eqref{eq:uncertainty}, we have one of the following possibilities:
\begin{align}\abs{\bra{\psi}\overline{X}\ket{\psi}}\leq \frac{1}{\sqrt{2}}, \quad \abs{\bra{\psi}\overline{Z}\ket{\psi}}\leq \frac{1}{\sqrt{2}}.\end{align}
Without loss of generality, assume the first holds.  Define the state $\ket{\theta}=\overline{X}\ket{\psi}$. Thus, $F(\psi, \theta)\leq \frac{1}{\sqrt{2}}$. We have the following claim, proved later.
\begin{claim}
\label{clm:localdistinguishable}
For every integer $K<w$, there is a set $T$ (a subset of ancillas and code qubits) of size at most $K\cdot 2^t$ such that 
\begin{align}\frac{1}{2}\|\psi_{T}- \theta_{T}\|_1 \geq \frac{K}{2^{t+4}w}.\end{align}
\end{claim}
Setting $K=\frac{d'}{2^{t+1}}\geq 1$ (recall $d'\leq w$ and $2^{t+1}\leq d'$), we find a set $T$ of size $|T|\leq d'/2$ such that $\frac{1}{2}\|\psi_{T}- \theta_{T}\|_1 \geq \frac{d'}{2^{2t+5}w}$. On the other hand, by assumption $\frac{1}{2}\|\psi-\rho\|_1< \frac{d'}{2^{2t+6}w}$ which implies $\frac{1}{2}\|\theta-\overline{X}\rho \overline{X}\|_1< \frac{d'}{2^{2t+6}w}$. Since $\rho_{\code} \in D_{s}$, first part of the Fact \ref{fact:samemarginals} ensures that $\rho_{T}=(\overline{X}\rho \overline{X})_{T}$. This implies, via triangle inequality, that $\frac{1}{2}\|\psi_{T}-\theta_{T}\|_1< \frac{d'}{2^{2t+5}w}$, which is a contradiction. This completes the proof.
\end{proof_of}

\begin{proof_of}{Claim \ref{clm:localdistinguishable}}
The main idea is that low-depth states are uniquely determined by their marginals on $2^t$ qubits. We are given a state $\ket{\psi}=U\ket{0}^{\otimes m}$ and $\ket{\theta}=\overline{X}\ket{\psi}$. Consider the Hamiltonian 
\begin{align}H=\sum_{S\subset [m]}P_S \defeq U\br{\sum_{j=1}^m \ketbra{1}_j}U^{\dagger}
\end{align}
which is a sum of commuting projectors and each $|S|\leq 2^t$. The unique ground-state is $\ket{\psi}$ and the spectral gap is $1$. Define $\overline{P}_S=\id-P_S$. Let $\Uu$ be set of $P_S$ that overlap the support of $\overline{X}$. Number of such $P_S$ is $u$ where $w\leq u\leq 2^t w$. For any $P_S\notin U$ we have $[P_S, \overline{X}]=0$ which implies $\bra{\theta}P_S\ket{\theta}=\bra{\psi}P_S\ket{\psi}=0$. Consider the operator
\begin{align}W_K \defeq \id + \br{\frac{\sum_{S\in U}P_S}{u}-\id}^K ,\end{align}
where $K< s$ is odd. It holds that\footnote{To show this, consider $\Pi^{\perp}_U$ (the excited space of $\sum_{S\in U}P_S$) and $\Pi^{\perp}=\id-\ketbra{\psi}$ (the excited space of $H$). Since $P_S\ket{\theta}=0$ for all $S\notin U$, we have $\bra{\theta}\Pi^{\perp}_U\ket{\theta}=\bra{\theta}\Pi^{\perp}\ket{\theta}$. Next, we show that $W_k\succeq \frac{k}{2u}\Pi^{\perp}_U$. For this, we need to argue that $1+\br{\frac{v}{u}-1}^k \geq \frac{k}{2u}$ for all $v\geq 1$. This is trivial when $v\geq u$. For $1\leq v < u$, consider 
\begin{align}
(1-\frac{v}{u})^K\leq (1-\frac{1}{u})^K\leq e^{-\frac{K}{u}}\leq 1-\frac{K}{2u}.
\end{align}Here we used $K<w\leq u$.} \begin{align}\bra{\theta}W_K\ket{\theta} \geq \frac{K}{2u}\bra{\theta}(\id - \ketbra{\psi}{\psi})\ket{\theta}.\end{align}
Since $F(\psi,\theta)\leq \frac{1}{\sqrt{2}}$, we have $\bra{\theta} (\id-\ketbra{\psi}{\psi}) \ket{\theta} \geq 1-\frac{1}{2}= \frac{1}{2}$. Thus, 
\begin{align}\bra{\theta}W\ket{\theta} \geq \frac{K}{2u}\bra{\theta} (\id-\ketbra{\psi}{\psi}) \ket{\theta} \geq \frac{K}{4u}.\end{align}
On the other hand (using $P_S\ket{\theta}=0$ for $S\notin \Uu$), 
\begin{xalign}
\bra{\theta}W\ket{\theta}&= 1+\bra{\theta}\br{\frac{\sum_{S\in U}P_S}{u}-\id}^K\ket{\theta}\\
&=1-\bra{\theta}\br{\frac{\sum_{S\in U}\overline{P}_S}{u}}^K\ket{\theta}\\
&=\frac{1}{u^K}\sum_{S_1, \ldots S_K\in U}\bra{\theta}\br{\id-\overline{P}_{S_1}\overline{P}_{S_2}\ldots \overline{P}_{S_K}}\ket{\theta}.
\end{xalign}
 We conclude that there exist $S_1,S_2, \ldots S_K\in \Uu$ such that for $S' =  S_1 \cup S_2 \cup \ldots \cup S_K$,
\begin{align}\Tr(\br{\id-\overline{P}_{S_1}\overline{P}_{S_2}\ldots \overline{P}_{S_K}}\theta_{S'}) \geq \frac{K}{4u}\geq \frac{K}{2^{t+4}w}.
\end{align}
The size of $S'$ is at most $K\cdot 2^t$.
Using 
\begin{align}\Tr(\br{\id-\overline{P}_{S_1}\overline{P}_{S_2}\ldots \overline{P}_{S_K}}\psi_{S'})=0,
\end{align} we get $\frac{1}{2}\|\theta_{S'}-\psi_{S'}\|_1\geq \frac{K}{2^{t+4}w}$.
\end{proof_of}

\section{Proof of improved circuit lower bounds using AGSPs}
\label{sec:appendix-pf}

\begin{proof_of}{Lemma \ref{lem:rootdwarmup}}
Suppose $2^t\geq d$, then the proof is immediate. Thus, assume $2^t< d$. Using Uhlmann's theorem \cite{uhlmann76}, we find that there is a purification $\ket{\rho_0}$ of $\rho_0$ such that 
\begin{align}F(\ketbra{\psi},\ketbra{\rho_0})=F(\psi_{\code},\rho_0)= f.
\end{align}
Since $\ket{\psi}$ has the maximum fidelity (over all code-states) with the projected vector 
\begin{align}\ket{\rho} \defeq \frac{1}{\sqrt{\bra{\psi}\Pi_\Cc\ket{\psi}}}\Pi_\Cc\ket{\psi},\end{align} we conclude $F(\ketbra{\psi},\ketbra{\rho})\geq f$. Defining the rotated vector $\ket{\rho'}=U^{\dagger}\ket{\rho}$, we further find that $F(\ketbra{0}^{\otimes m}, \rho')\geq f$.

Since the vector $\ket{\rho'}$ is the projection of $\ket{0}^{\otimes m}$ on the rotated code-space $\Cc'\defeq U^{\dagger}(\id\otimes \Cc) U$, it is expected to have low entanglement. This is formalized below, adapted from \cite{hastings2007area_law, AradLV12, AradKLV13} and proven later.
\begin{claim}
\label{clm:quadlowent}
For any region $R\subset [m]$, it holds that
\begin{align}S(\rho'_R) \leq 32\left(\sqrt{2^{t+1}\ell|R|\log\frac{1}{f}}\right)\log^2 (2^{t+1}\ell|R|).\end{align} 
\end{claim}

Define the state $\Theta$ from $\ket{\rho}$ as in \eqref{eq:maxmixcode} and let $\Theta'= U^{\dagger}\Theta U$. It holds that for any region $R$ of size $<d$, $\Theta_R=\rho_R$. Fact \ref{fact:lightfact2} now implies that for any region $R'$ of size at most $\frac{d}{2^{t+1}}$, $\Theta'_R=\rho'_R$. Thus, dividing $[m]$ into $\frac{2^{t+1}m}{d}$ regions $R_1,R_2, \ldots R_{\frac{2^{t+1}m}{d}}$, each of size at most $\frac{d}{2^{t+1}}$, we find that 
\begin{align}k\leq S(\Theta') \leq \sum_{j=1}^{\frac{2^{t+1}m}{d}}S(\Theta'_{R_j})=\sum_{j=1}^{\frac{2^{t+1}m}{d}}S(\rho'_{R_j})\leq \frac{2^{t+1}m}{d}\cdot \left(8\sqrt{d\ell\log\frac{1}{f}}\right)\log^3 d\ell.\end{align}
We used Claim \ref{clm:quadlowent} above. Since $m\leq 2^tn$, we can re-write this as
\begin{align}
k \leq \frac{2^{2t+6}n\log^2 d\ell}{\sqrt{d}}\cdot \sqrt{\ell\log\frac{1}{f}}.\end{align}
Thus, we conclude that 
\begin{align}2^{2t}\geq \frac{k\sqrt{d}}{64 n \log^2 d\ell \cdot \sqrt{\ell\log\frac{1}{f}}}.\end{align}
This completes the proof.
\end{proof_of}

Now, we prove Claim \ref{clm:quadlowent}. It is a simple application of the Approximate Ground-State Projector (AGSP) framework based on polynomial approximations to local Hamiltonian \cite{AradKLV13}. We will use the following well known polynomials that improve upon the Chebyshev approximation to AND function.

\begin{fact}[\cite{KahnLS96, BuhrmanCWZ99}]
\label{fact:KLS}
Let $n$ be an integer and $h:\{0,1,\ldots n\}\rightarrow \{0,1\}$ be the function defined as $h(0)=1$ and $h(j)=0$ for $j\in [n]$. For every $\sqrt{n}\leq \dgre \leq n$, there is a polynomial $K_{\dgre}$ of degree $\dgre$ such that for every $j\in \{0,1,\ldots n\}$,  $|h(j)-K_{\dgre}(j)|\leq \exp\left(-\frac{\dgre^2}{2^8n}\right)$.
\end{fact}

\begin{proof_of}{Claim \ref{clm:quadlowent}}

Let $\Pi'_j\defeq U^{\dagger}\Pi_j U$ be the rotated commuting checks.  Each $\Pi'_j$ has locality $\leq\ell 2^t$ and each qubit participates in at most $\ell 2^t$ rotated checks. Let $R_1$ be the extended region defined as the set of all qubits that share a rotated check with a qubit in $R$. Let $\Pi_c = \times_{j: \supp(\Pi'_j)\notin R_1}\Pi_j$ be the common eigenspace of all the checks not in $R_1$ and define the ``truncated Hamiltonian'' 
\begin{align}
H_{\trunc}= \left(\sum_{j:\supp(\Pi'_j)\subset R_1}(\id- \Pi'_j)\right) + (\id-\Pi_c).
\end{align}
Note that $\Pi_{\Cc'}$ is the ground-space of $H_{\trunc}$ and the spectral gap of $H_{\trunc}$ is $1$. The advantage is that the norm of $H_{\trunc}$ is now $\leq 2^t\ell|R|+1 \leq 2^{t+1}\ell|R|$.
For an integer $\dgre$ to be chosen later, consider the degree $\dgre$ polynomial of $H_\trunc$ obtained from Fact \ref{fact:KLS}: $K_{\dgre}(H_{\trunc})$. It satisfies 
\begin{align}\|K_{\dgre}(H_{\trunc})-\Pi_{\Cc'}\|_{\infty} \leq \exp\left(-\frac{\dgre^2}{2^{t+9}\ell|R|}\right).
\end{align} 
This ensures that the state 
\begin{align}\ket{\omega} \defeq \frac{K_{\dgre}(H_{\trunc})\ket{0}^{\otimes m}}{\|K_{\dgre}(H_{\trunc})\ket{0}^{\otimes m}\|_1}
\end{align} satisfies \begin{align}\norm{\ket{\omega}-\ket*{\rho'}}_1 \leq \frac{2}{f} \cdot \exp \left(-\frac{\dgre^2}{2^{t+9}\ell|R|}\right).
\end{align}
Letting $\dgre= \sqrt{2^{t+9}\ell|R|\log\frac{2|R|}{f}}$, we conclude that 
$\|\ket{\omega}-\ket{\rho'}\|_1\leq \frac{1}{|R|}$. Claim \ref{clm:schmidtrank}, below, shows that the Schmidt rank of $K_{\dgre}(H_{\trunc})$ across $R$ and $[m]\setminus R$ is at most $(2^{2t+1}\ell^2|R|)^{\dgre}$. Thus,
\begin{align}S(\omega_R)\leq \dgre\cdot \log (2^{2t+1}\ell^2|R|)\leq 2\dgre\cdot\log (2^{t+1}\ell|R|).\end{align} Using the Alicki-Fannes inequality \cite{AlickiF04}, we thus find that 
\begin{xalign}
S(\rho'_R) &\leq 2|R|\cdot \frac{1}{|R|} + S(\omega_R) \\ &\leq 2\dgre\cdot \log (2^{t+1}\ell|R|)+2 \\
&\leq 2\left(\sqrt{2^{t+9}\ell|R|\log\frac{2|R|}{f}}\right)\log(2^{t+1}\ell|R|) +2\\
&\leq 32\left(\sqrt{2^{t+1}\ell|R|\log\frac{1}{f}}\right)\log^2 (2^{t+1}\ell|R|).
\end{xalign}
This completes the proof.
\end{proof_of}

\begin{claim}
\label{clm:schmidtrank}
Schmidt rank of any degree $\dgre$ polynomial $K_{\dgre}(H_{\trunc})$ across $R$ and $[m]\setminus R$ is at most 
\begin{align}\dgre^3\cdot (2^{2t}\ell^2|R|)^{\dgre}\leq (2^{2t+1}\ell^2|R|)^{\dgre}.
\end{align}
\end{claim}
\begin{proof}
Let us provide an upper bound on the Schmidt rank of $(H_{\trunc})^q$, for any $0\leq q \leq \dgre$. Let $H_{\trunc}= H_{\partial} + H_{\mathsf{in}} + H_{\mathsf{out}}$, where $H_{\partial}$ is the set of all rotated checks supported on both $R$ and $[m]\setminus R$, $H_{\mathsf{in}}$ is the set of rotated checks strictly within $R$ and $H_{\mathsf{out}}$ is the set of rotated checks (including the truncated part $\Pi_c$) within $[m]\setminus R$. Note that all these terms commute and the number of $H_{\partial}$ is at most $2^t\ell|R|$ (since each qubit in $R$ participates in at most $2^t\ell$ rotated checks). Then $(H_{\trunc})^q= \sum_{a+b+c=q}H_{\mathsf{in}}^a H_{\mathsf{out}}^b H_{\partial}^c$. The operators $H_{\mathsf{in}}^a$ and $H_{\mathsf{out}}^b$ do not increase the Schmidt rank across $R$ and $[m]\setminus R$. The Schmidt rank of $H_{\partial}$ is $\leq 2^t\ell|R|\cdot 2^t\ell = 2^{2t}\ell^2|R|$, since each of the $2^t\ell|R|$ rotated checks has Schmidt rank $2^t\ell$. Thus, $(H_{\trunc})^q$ has Schmidt rank 
\begin{align}\leq q^2\cdot (2^{2t}\ell^2|R|)^q\leq \dgre^2\cdot (2^{2t}\ell^2|R|)^{\dgre}.\end{align} Finally, $K_{\dgre}(H_{\trunc})$ has Schmidt rank at most $\dgre$ times this number. This completes the proof.
\end{proof}

\section{Amplification of circuit lower bounds}

\label{sec:amplification}

In the previous sections, we considered local Hamiltonians $H$ which were the sum $\sum_{i = 1}^N H_i$ of local terms $H_i$ of norm $\leq 1$. In this framework, we were interested in the circuit complexity of states of energy $\leq \eps N$. In this section, we will shift to an equivalent framework and let $H = \Exp_i H_i$ and consider states of energy $\leq \eps$. 
This is because we will be considering Hamiltonians of super constant locality and it is notationally simpler to consider normalized Hamiltonians.
We define the locality of a Hamiltonian as follows.

\begin{definition}
We say a local Hamiltonian is $(\ell, D)$-local if each Hamiltonian term acts non-trivially on at most $\ell$ qubits and each qubit is acted on non-trivially by at most $D$ Hamiltonian terms. We will also refer to a Hamiltonian as simply $\ell$-local if it is $(\ell, \ell)$-local.
\end{definition}
The main result of this section is a simple transformation one can apply to a local Hamiltonian instance to improve circuit depth lower bounds at the cost of worsening the locality of the Hamiltonian. The principal idea is to transform the Hamiltonian $H$ into a Hamiltonian $H'$ such that for any low-depth state $\phi$, 
\begin{align}\tr(H' \phi) \geq p \cdot \tr(H \phi)
\end{align}
for some choice of $p > 1$. If we can construct such an energy amplification, then any depth lower bound we had for \emph{very} low-energy states of $H$ can be translated to a depth lower bound for low-energy states of $H$. For example, in the case of the lower bounds proven in the prior section which scale roughly as $\log (1/\eps)$ for states of energy $\leq \eps $, applying such a transformation would give us a lower bound of $\log(p/\eps)$ for states of energy $\leq \eps$.
The following theorem shows how to achieve this result for low depth states at a cost of increasing the locality of the Hamiltonian by a factor of $p$. 

\begin{theorem}
Let $H$ be a $\ell$-local Hamiltonian on $n$ qubits such that $H \geq 0$. Define the Hamiltonian $H^{(p)}$ as
\begin{align}
    H^{(p)} \defeq  \II - \left( \II - H \right)^p.
\end{align}
Then for any mixed state $\phi$ of $\cc(\phi) \leq t$, we have
\begin{align}
    \tr(H^{(p)} \phi) \geq \frac{1}{2} \min \left\{1, p \tr(H \phi ) \right\} - \frac{2^t p^2\ell^2}{n}.
\end{align}
\label{thm:amplification}
\end{theorem}
We prove this theorem at the end of the section. We now apply this theorem to our previous lower bounds to generate a super-constant locality Hamiltonian with super-constant circuit lower bounds for all states of energy $\leq 1/100$ with respect to the new Hamiltonian. The following is a reformulation of Theorem \ref{thm:main} (Theorem \ref{thm:main-technical}) applied to stabilizer codes of linear rate and polynomial distance such as the Tillich-Z\'emor code \cite{5205648}.

\begin{corollary}
For fixed constants $\ell > 2, c > 0$, there exists a family of $\ell$-local Hamiltonians $H$ on $n$ qubits such that for any state $\phi$ of energy $\tr(H\phi) \leq \eps$, the circuit complexity of $\phi$ is at least
\begin{align}
    \cc(\phi) \geq c \cdot \min \left\{ \log n, \log \frac{1}{\eps} \right\}.
\end{align}
\label{cor:reformulation}
\end{corollary}
To apply Theorem \ref{thm:amplification}, let us consider a state $\phi$ of circuit complexity $\cc(\phi) \leq t < \frac{1}{3} \log n - \log 100 \ell^3$ (otherwise, the lower bound is trivial) and an amplification factor of $p \leq n^{1/3}$. Let us also assume that $\tr(H^{(p)} \phi) \leq \frac{1}{100}$. If $\tr(H \phi) \geq \ell / p$, then we reach a contradiction as
\begin{align}
    \frac{1}{100} \geq \tr(H^{(p)} \phi) \geq \frac{\ell}{2} - \frac{p^2 n^{1/3}}{100n} \geq \frac{\ell}{2} - \frac{1}{100} \geq 1 - \frac{1}{100}.
\end{align}
Therefore, we may assume $\tr(H \phi) < \ell / p$. Then,
\begin{align}
    \frac{1}{100} \geq \frac{p}{2} \tr(H \phi) - \frac{p^2 n^{1/3}}{100n} \geq \frac{p}{2} \tr(H \phi) - \frac{1}{100},
\end{align}
or equivalently,
\begin{align}
    \tr(H \phi) \leq \frac{1}{25 p}.
\end{align}
Therefore, the circuit complexity of $\phi$ is at least
\begin{align}
    \cc(\phi) \geq c \cdot \min \left\{ \log n, \log \frac{p}{25} \right\},
\end{align}
thus proving a circuit lower bound for all states of energy $\leq 1 / 100$ with respect to $H^{(p)}$. The locality of the Hamiltonian $H^{(p)}$ can be calculated as follows: Each term of $H^{(p)}$ is a product of $p$ terms of $H$ and therefore acts on $p \ell$ qubits. Likewise, each qubit of $H^{(p)}$ is acted on by $\leq p \ell^{p+1} n^{p-1} \defeq D$ terms since there are $\leq (\ell n)^p$ Hamiltonian terms in $H^{(p)}$. 

Assuming we are content with each Hamiltonian term acting on $p \ell$ qubits, we can apply the operator Chernoff bound to sparsify the Hamiltonian such that each qubit does not act in too many terms. To sparsify the Hamiltonian $H^{(p)}$, we select a uniformly random subset of $k$ terms from the Hamiltonian $H^{(p)}$ and consider the Hamiltonian $H'$ defined as the expectation over these $k$ terms. The following lemma demonstrates that with high probability over this sparsification procedure, the spectra of $H^{(p)}$ and $H'$ are close.

\begin{lemma}
Fix $\delta > 0$ and consider a normalized Hamiltonian $(\ell, D)$-local Hamiltonian $G$ on $n$ qubits. For a choice of 
\begin{align}
    k = n \cdot \max \left\{ \frac{32}{\delta^2}, \frac{\log n}{\ell} \right\}
\end{align}
the sparsified Hamiltonian $G'$ constructed by choosing $k$ random terms of $G$ satisfies the following: With probability $\geq \frac{1}{3}$ over the sparsification,
\begin{align}
    \norm{G - G'} \leq \delta
\end{align}
and each term of $G'$ participates in at most $O(\ell/\delta^2)$ Hamiltonian terms.
\label{lem:sparsification}
\end{lemma}

The proof of this lemma is given at the end of this section. Directly applying the lemma on $H^{(p)}$ for $\log n < p < n^{1/3}$ using $\delta = 1/200$, we get that there exists a local Hamiltonian $H'$ consisting of $\Theta(n)$ local terms with each term of $H'$ participates in $O(p)$ Hamiltonian terms. Furthermore, $H'$ is $1/200$-close to $H^{(p)}$ in spectral norm, so any state of energy $\leq 1/200$ with respect to $H'$ has energy $\leq 1/100$ with respect to $H^{(p)}$ and the previously proven circuit lower bound applies. Restated, we achieve the following.

\begin{theorem}[Super constant locality NLTS]
For $p(n)$ a function such that $\log n < p(n) < n^{1/3}$, there exists a $O(p(n))$-local family of local Hamiltonians acting on $n$ qubits and consisting of $N = \Theta(n)$ local terms such that every state $\phi$ of energy $\leq N/200$ has a circuit complexity lower bound of
\begin{align}
    \cc(\phi) \geq \Omega(\log p(n)).
\end{align}
\label{thm:amp-restate}
\end{theorem}

\subsection{Developing locality reduction}

By itself, Theorem \ref{thm:amp-restate}, is not a surprising result since the lower bound we get is the logarithm of the locality of the Hamiltonian. Consider the Hamiltonian
\begin{align}
    H = \frac{p}{n} \left(\Pi_1 + \Pi_2 + \ldots + \Pi_{n/p} \right) \label{eq:cat-counterex}
\end{align}
where each $\Pi_i$ is the projector onto the state $\ket{\cat} = \frac{1}{\sqrt{2}} (\ket{0}^{\otimes p} + \ket{1}^{\otimes p})$ for a disjoint set of $p$ qubits. \cite{nirkhe_et_al:LIPIcs:2018:9095} show how to prove a circuit lower bound of $\Omega(\log p)$ for all states of energy $\leq 1/200$ with respect to $\Pi_i$. Using a simple Markov inequality, one can extend it to a circuit lower bound for the low-energy states of $H$ from \eqref{eq:cat-counterex}. Furthermore, a small variation of this Hamiltonian can be constructed with $\Theta(n)$ local terms.

What Theorem \ref{thm:amp-restate} provides, however, is a technique for amplifying the circuit complexity of a Hamiltonian at the cost of increasing locality. Consider then a hypothetical locality reducing transformation which takes as input a Hamiltonian $H$ and outputs a $\ell$-local Hamiltonian $H'$, for a fixed constant $\ell$, such that for  any low-depth state $\phi$, $\tr(H' \phi) \geq c \cdot \tr(H \phi)$ for some fixed constant $c > 0$. Then one could sequentially apply amplification and locality reduction transformations until an NLTS Hamiltonian was constructed. This is analogous to Dinur's construction of the PCP theorem \cite{dinur-pcp}. 

Unfortunately, the locality reduction step of Dinur's proof relies heavily on copying information, a luxury unavailable in the quantum setting. In fact, it is unclear if we should even believe that a locality reduction transformation exists. For one, a Hamiltonian term testing a global property cannot be replicated by any family of local Hamiltonians. For example, the $\ket{\cat}$ state is not the ground-state of any local Hamiltonian and therefore the projector $\ketbra{\cat}$ cannot be approximated by local terms\footnote{However, a state near the $\ket{\cat}$ state is the unique ground-state of a local Hamiltonian as shown by Nirkhe, Vazirani and Yuen \cite{nirkhe_et_al:LIPIcs:2018:9095} .}.

\subsection{Omitted Proofs}

\begin{proof_of}{Theorem \ref{thm:amplification}}
Let $H = \Exp_i h_i$ being the decomposition into local terms, $g_i = \II - h_i$, and $G = \Exp_i g_i$ so $G = \II - H$ and $H^{(p)} = \II - G^p$.
Let $\phi$ be a state of depth $\cc(\phi) \leq t$. Then, we can express the energy of $\phi$ with respect to $H_p$ as
\begin{xalign}
    \tr(H^{(p)} \phi) &= 1 - \tr \left( G^p \phi \right) \\
    &= 1 - \left( 1 - \tr(H \phi) \right)^p - \left[ \tr(G \phi)^p - \tr(G^p \phi) \right].
\end{xalign}
To bound this difference of terms, we will use the property that $\phi$ has a low-depth circuit. We can write
\begin{xalign}
    \tr(G \phi)^p - \tr(G^p \phi) &= \Exp_{i_1, \ldots, i_t} \left( \tr(g_{i_1} \ldots g_{i_t} \phi) - \tr(g_{i_1} \phi) \ldots \tr( g_{i_t} \phi) \right) \\
    &\leq \Pr_{i_1, \ldots, i_t} \left( \tr(g_{i_1} \ldots g_{i_t} \phi) \neq \tr(g_{i_1} \phi) \ldots \tr( g_{i_t} \phi) \right)
\end{xalign}
since each Hamiltonian term $g_i$ is normalized. We claim that the only sequences $(i_1, \ldots, i_t)$ for which $\tr(g_{i_1} \ldots g_{i_t} \phi)$ may not equal $\tr(g_{i_1} \phi) \ldots \tr(\phi g_{i_t})$ are those for which $g_{i_k}$ falls in the light cone of a previous $g_{i_1}, \ldots, g_{i_{k-1}}$. This is because the value of an observable only depends on the reduced state in its light cone. Assume $g_{i_1}, \ldots, g_{i_t}$ have disjoint light cones $L_1, \ldots, L_t$. Then
\begin{xalign}
\tr(g_{i_1} \ldots g_{i_t} \phi) &= \tr \left( g_{i_1} \ldots g_{i_t} \phi_{L_1 \cup \ldots \cup L_t} \right) \\
&= \tr \left( g_{i_1} \ldots g_{i_t} \phi_{L_1} \otimes \ldots \otimes \phi_{L_t} \right) \\
&= \tr(g_{i_1} \phi_{L_1}) \ldots \tr(g_{i_t} \phi_{L_t}) \\
&= \tr(g_{i_1} \phi) \ldots \tr(g_{i_t} \phi).
\end{xalign}
Therefore, we can upper bound the probability of this event by the probability that $g_{i_k}$ does not fall in the light cone of any previous $g_{i_1}, \ldots, g_{i_{k-1}}$ for all $k$. We use the fact each light cone has size at most $\ell 2^t$ and that $g_{i_k}$ has size $\ell$ and apply a union bound:
\begin{align}
    \left( 1 - \frac{2^d \ell^2}{n} \right) \cdot \left( 1 - 2 \frac{2^d \ell^2}{n} \right) \cdot \ldots \left( 1 - (t-1) \frac{2^d \ell^2}{n} \right) \geq 1 - \frac{2^d t^2 \ell^2}{n}.
\end{align}
Lastly, we combine this bound with trivial bound for $1 - \left( 1 - x \right)^p $ of $\min \left\{1, p x\right\} / 2$ to achieve
\begin{align}
\tr(H^{(p)} \phi) \geq \frac{\min\left\{1, p \tr(H \phi)\right\}}{2} - \frac{2^d t^2 \ell^2}{n}.
\end{align}
\end{proof_of}

\begin{proof_of}{Lemma \ref{lem:sparsification}}
Let $G$ be a $(\ell, D)$-local Hamiltonian where $G = \frac{1}{m} \sum_{i = 1}^m g_i$. Pick $i_1, \ldots, i_k$ uniformly randomly from $[m]$ and let $G' = \frac{1}{k} \sum_{j = 1}^k g_{i_j}$. Let $X$ be the operator-valued random variable which takes on the value $g_i$ with probability $\frac{1}{m}$. Then $G'$ is generated by $k$ samples from $X$. Applying the operator Chernoff bound \cite[Lemma 2.8]{tropp},
\begin{align}
    \Pr\left[\|H'-H\|\geq \eps\right]\leq 2^n e^{-k \delta^2/32} \leq \frac{1}{3}
\end{align}
using that $k \geq 32n/\delta^2$. Letting $Y_a$ be the random variable that equals 1 if the random term chosen by $X$ acts non-trivially on the qubit $a$ and 0 otherwise. Then $\Exp Y_a = \frac{D}{m} = \Theta(\frac{\ell}{n})$. Then for $k$ independent draws of $Y_a$: $\{Y_a^1, \ldots, Y_a^k\}$, the standard Chernoff bound states that
\begin{align}
    \Pr \left[ Y_a^1 + \ldots + Y_a^k \geq \Theta \left( \frac{k\ell}{n}\right) \right] \leq e^{-\Theta(k \ell/n)}.
\end{align}
Applying a union bound gives the probability that any qubit $a$ is acted on by more than $\Theta(k \ell/n)$ terms is at most $n e^{-\Theta(k\ell/n)}$. Since $k \geq (n \log n) / \ell$, this is at most $1/3$.
\end{proof_of}

\bibliography{references}
\bibliographystyle{alpha}

\end{document}